\documentclass[sigconf, natbib=true, review=false, authorsversion=True, noacm=True]{acmart}
\usepackage[utf8]{inputenc}
\usepackage{pdfpages}
\usepackage{svg}
\graphicspath{ {./graphics/} }
\usepackage{placeins}
\usepackage{subfig}
\usepackage[hide]{chato-notes}
\usepackage{enumitem}
\usepackage{bbm}
\usepackage{flushend}
\newtheorem{theorem}{Theorem}[section]
\newtheorem{corollary}{Corollary}[theorem]
\usepackage{multirow}
\usepackage[para]{footmisc}
\usepackage{array}
\usepackage{wrapfig}
\newcolumntype{H}{>{\setbox0=\hbox\bgroup}c<{\egroup}@{}}

\usepackage{amsmath}

\DeclareMathOperator*{\argmin}{arg\,min}

\usepackage{graphicx}
\usepackage{makecell}
\def\hlinewd#1{%
\noalign{\ifnum0=`}\fi\hrule \@height #1 %
\futurelet\reserved@a\@xhline} 

\DeclareMathOperator*{\E}{\mathbb{E}}
\DeclareMathOperator*{\I}{\mathbb{I}}
\DeclareMathOperator*{\Loss}{\mathcal{L}}
\DeclareMathOperator*{\LossBCE}{\mathcal{L}_{BCE}}

\DeclareMathOperator*{\LossGBCE}{\ensuremath{\mathcal{L}^{\beta}_{gBCE}}}

\newcommand\LossGBCEb[1]{\mathcal{L}^{#1}_{gBCE}}
\usepackage{marginnote}

\usepackage{ifthen}

\makeatletter
\newcommand{\conditionalnewline}[1]{%
  \ifthenelse{\equal{#1}{manuscript}}{%
    \\%
  }{%
    \ifthenelse{\equal{#1}{sigconf}}{%
      \\%
    }{}%
  }%
}

\usepackage{amsthm}
\renewcommand{\qedsymbol}{$\blacksquare$}
\makeatletter
\renewcommand{\qed}{%
  \ifmmode
    \mathqed
  \else
    \leavevmode\unskip\penalty9999 \hbox{}\nobreak\hfill
    \quad\hbox{\qedsymbol}
  \fi
}
\makeatother

\makeatother

\newcommand{\ourmodel}{gSASRec}
\newcommand{\ourloss}{gBCE}

\definecolor{brownweb}{rgb}{0.65, 0.16, 0.16}

\newcommand{\dd}[1]{\textcolor{black}{#1}}
\newcommand{\sdd}[1]{\textcolor{black}{#1}}

\newcommand{\scny}[1]{\textcolor{black}{#1}}
\newcommand{\swtd}[1]{\textcolor[HTML]{000000}{#1}}
\newcommand{\sfdtd}[1]{\textcolor[HTML]{000000}{#1}}
\newcommand{\stdtd}[1]{\textcolor[HTML]{000000}{#1}}
\newcommand{\sm}[1]{\textcolor[HTML]{000000}{#1}}
\newcommand{\sd}[1]{\textcolor[HTML]{000000}{#1}}
\newcommand{\craig}[1]{\textcolor{black}{#1}}
\newcommand{\shd}[1]{\textcolor{black}{#1}}
\newcommand{\srs}[1]{\textcolor[HTML]{000000}{#1}}
\newcommand{\srsf}[1]{\textcolor[HTML]{000000}{#1}}
\newcommand{\srsd}[1]{\textcolor[HTML]{000000}{#1}}
\newcommand{\rsm}[1]{\textcolor[HTML]{000000}{#1}}
\newcommand{\scrc}[1]{\textcolor[HTML]{000000}{#1}}
\newcommand{\scrci}[1]{\textcolor[HTML]{000000}{#1}}

\copyrightyear{2023} 
\acmYear{2023} 
\setcopyright{acmlicensed}\acmConference[RecSys '23]{Seventeenth ACM Conference on Recommender Systems}{September 18--22, 2023}{Singapore, Singapore}
\acmBooktitle{Seventeenth ACM Conference on Recommender Systems (RecSys '23), September 18--22, 2023, Singapore, Singapore}

\acmDOI{}
\acmPrice{}
\acmISBN{}
\setcopyright{rightsretained}

\author{Aleksandr Petrov}
\affiliation{%
  \institution{University of Glasgow} \country{United Kingdom}}

\email{a.petrov.1@research.gla.ac.uk}

\author{Craig Macdonald}
\affiliation{%
  \institution{University of Glasgow} \country{United Kingdom}}
\email{craig.macdonald@glasgow.ac.uk}


\title{\ourmodel{}: Reducing Overconfidence in Sequential Recommendation Trained with Negative Sampling}

\begin{document}
\begin{abstract}
\looseness -1 \scrci{A large} catalogue size is one of the central challenges in training recommendation models: a large number of items makes \scrci{them} \srsf{memory and computationally inefficient}  to compute scores for all items during training,  forcing \scrci{these} models to deploy negative sampling. However, negative sampling increases the proportion of positive interactions in the training data, and therefore models trained with negative sampling \dd{tend} to overestimate the probabilities of positive interactions -- a phenomenon we call {\em overconfidence}. While the absolute values of the predicted scores/probabilities are \scrci{not important} for the ranking of retrieved recommendations, overconfident models may fail to estimate nuanced differences in the top-ranked items, resulting in degraded performance. \scrci{In this paper, we show that overconfidence explains why the popular SASRec model underperforms when compared to BERT4Rec. This is contrary to the BERT4Rec authors' explanation that the difference in performance is due to the bi-directional attention mechanism}. To mitigate overconfidence, we propose a novel Generalised Binary Cross-Entropy Loss function (gBCE) and theoretically prove that it can mitigate overconfidence. We further propose the gSASRec model, an improvement over SASRec that deploys an increased number of negatives and \scrci{the} gBCE loss. We show through detailed experiments on three datasets that gSASRec does not exhibit the overconfidence problem. As a result, gSASRec can outperform BERT4Rec (e.g.\ +9.47\% NDCG on \scrci{the} MovieLens-1M \scrci{dataset}), while requiring less training time (e.g.\ -73\% training time on MovieLens-1M). Moreover, in contrast to BERT4Rec, gSASRec is suitable for large datasets that contain more than 1 million items.
\end{abstract}

\maketitle

\section{Introduction}\label{sec:intro}
\looseness -1 \emph{Sequential Recommender Systems} is a class of recommender systems that take into account the order of user-item interactions and aim to predict the next item in a sequence. Taking order into account is important in many \dd{recommendation scenarios} -- for example, if a user just bought a mobile phone, then the next purchase is likely to be an accessory for this phone, and hence it makes sense to recommend such accessories \scrci{to} the user.  Recently, models based on architectures developed for natural language modelling (particularly using Transformers~\cite{Transformer}) have achieved state-of-the-art performance in sequential recommendation~\cite{SASRec, BERT4Rec, DuoRec, Bert4RecRepro, PetrovRSS22, petrov2023rss}. The success of these language model architectures for sequential recommendation is explained by the similarities between \dd{modelling} sequences of words in texts and \dd{modelling} sequences of user-item interactions.
However, a direct adaptation of language model architectures for sequential recommendation \scrci{can} be problematic because the number of items in the \dd{system} catalogue can be much larger than the corresponding vocabulary size of the language models. For example, YouTube has a catalogue of more than 800 million videos\footnote{\href{https://earthweb.com/how-many-videos-are-on-youtube/}{https://earthweb.com/how-many-videos-are-on-youtube/}}. In practice, a direct adaptation of language models with catalogue sizes exceeding 1 million items is computationally prohibitive~\cite{PetrovRSS22, petrov2023rss}. 
\srsd{Indeed, compared with traditional matrix factorisation models that compute one score distribution per user, sequential recommendation models are usually trained to predict scores for each position in the sequence, meaning that the model has to generate $S*N$ scores per sequence, where $S$ is the sequence length, and $N$ is the size of the catalogue.} \srsd{For example, to train a sequential recommendation \scrci{model} with 64 sequences of 200 items per batch, \scrci{having} 1M \rsm{items}, would require ~51GB GPU memory without \scrc{accounting for \scrci{the} training} gradients. Factoring in \scrci{the} gradients and model weights increase this to more than 100Gb, \scrci{thereby exceeding consumer-grade GPU capacities.}} 

\inote{R1 says negative sampling is not a full solution, as it does not speed up inference. The motivation is mostly training-related and mostly about memory issues rather than speed}

A typical solution for this \scrc{training} problem is \emph{negative sampling}: the models are trained on all \emph{positive} interactions (the user-item interactions present in the training set) but only sample a very small fraction of 
\emph{negative} interactions (all other possible \sfdtd{but unseen} user-item interactions). 
\looseness -1 Negative sampling is known to be one of the central challenges~\cite{rendleItemRecommendationImplicit2022} in training recommender systems: it increases the proportion of positive samples in the training data distribution, and therefore models learn to overestimate the probabilities of future user-item interactions. We describe this phenomenon as \emph{overconfidence}. While the magnitude of retrieval scores is typically \scrci{non-important} for the {\em ranking} of items, overconfidence is problematic %
because models frequently fail to focus on nuanced variations in the highly-scored items and focus on distinguishing top vs.\ bottom items instead. \sm{Another problem caused by overconfidence is more specific to models trained with the popular Binary Cross-Entropy loss: if an item $i$ with high predicted probability $p_i$ is sampled as a negative, $\log(1-p_i)$ calculated by the loss function tends to $-$infinity, causing numerical overflows and unstable training. Overall, we argue that over\-co\-nfi\-dence hinders model effectiveness and makes model training hard.}

\srs{Although overconfidence is a general problem applicable to \emph{all} recommender systems trained with negative sampling, in this paper, we focus specifically on \emph{sequential} recommender systems, for which negative sampling is specifically important, due to the large GPU memory requirement discussed above. Indeed, as we show in this paper, the use of negative sampling leads to overconfidence in the popular SASRec~\cite{SASRec} sequential recommendation model. Existing solutions \dd{that can address} \sm{overconfidence induced by} negative sampling in recommender systems (e.g.~\cite{rendleImprovingPairwiseLearning2014, yuanLambdaFMLearningOptimal2016}) are hard to adapt to deep learning-based sequential recommender models (see also Section~\ref{ssec:neg_sampling}). \scrci{Hence,} the overconfidence \scrc{issue present} in negatively-sampled sequential recommendation models remains largely unsolved. \scrc{Indeed}, the state-of-the-art BERT4Rec~\cite{BERT4Rec} \dd{model} does not use negative sampling and, therefore, cannot be applied to datasets with large catalogues.\footnote{\scrc{By BERT4Rec, we refer to the model architecture, the training task and the loss function. As we show in Section~\ref{sssec:results:bert4rec_performance}, \scrci{while it is possible to train BERT4Rec's architecture while using negative sampling, doing so negatively impacts the model's effectiveness.}}}}

\looseness -1 \srsf{Hence, to address the overconfidence issue in the sequential recommendation, we introduce a novel Generalised Binary Cross-Entropy loss (\ourloss{}) – a generalisation of BCE loss using a generalised logistic sigmoid function~\cite{richardsFlexibleGrowthFunction1959, phamCombinationAnalyticSignal2019}. We further propose the Generalised SASRec model (\ourmodel{}) – an enhanced version of SASRec~\cite{SASRec} trained with more negative samples and \ourloss{}. Theoretically, we prove that \ourmodel{} can avoid overconfidence even when trained with negative sampling (see Theorem~\ref{theorem:bias}). Our theoretical analysis aligns with  \scrc{an} empirical evaluation of \ourmodel{} on three datasets (Steam, MovieLens-1M, and Gowalla), demonstrating the benefits of \scrci{having} more negatives and \scrci{the} \ourloss{} loss during training. On smaller datasets (Steam and MovieLens-1M), the combination of these improvements significantly outperforms BERT4Rec's performance on MovieLens-1M (+9.47\% NDCG@10) and achieves comparable results on Steam (-1.46\% NDCG@10, not significant), \scrc{while} requiring much less time to converge. Additionally, \ourloss{} shows benefits when used with BERT4Rec trained with negative samples (+7.2\% \srsd{NDCG@10} compared with BCE Loss on MovieLens-1M with 4 negatives). On the Gowalla dataset, where BERT4Rec training is infeasible due to large catalogue size~\cite{PetrovRSS22, petrov2023rss}, we obtain substantial improvements over the regular SASRec model (+47\% NDCG@10, statistically significant). Although this paper focuses on sequential recommendation, our proposed methods and theory could be applicable to other research areas, such as recommender systems (beyond sequential recommendation), \scrci{search systems}, or natural language processing.}

\looseness -1 \dd{In short, our contributions can be summarised as follows: (i) we define overconfidence through a probabilistic interpretation of sequential recommendation; (ii)  we show (theoretically and empirically) that SASRec is prone to overconfidence due to its negative sampling; (iii) we propose \ourloss{}~loss and theoretically prove that it can mitigate the overconfidence problem; (iv) we use \ourloss{}~to train \ourmodel{}~and show that it exhibits better (on MovieLens-1M) or similar (on Steam) effectiveness to BERT4Rec, while \scrci{both} requiring less training time, and also being suitable for training on large datasets.}

The rest of this paper is as follows: Section~\ref{sec:related_work} provides an overview of related work; Section~\ref{sec:background} formalises sequential recommendation and \scrci{the} typically used loss functions; we describe the problem of overconfidence in Section~\ref{sssec:overconfidence}; in Section~\ref{sec:gbce_and_properties} we introduce \ourloss{}~and theoretically analyse its properties, before defining \ourmodel{}; Section~\ref{sec:experiments} experimentally analyses the \scrci{impact} of negative sampling in \sfdtd{SASRec, BERT4Rec and~\ourmodel{}}; Section~\ref{sec:conclusion} \scrci{provides} concluding remarks. 

\section{Related Work}\label{sec:related_work}
In this section we discuss existing work related to negative sampling in recommender systems.  We review existing approaches for traditional (Matrix Factorisation-based) recommender systems in Section~\ref{ssec:neg_sampling} and discuss why they are hard to apply for sequential recommendation. \dd{We} then discuss training objectives \inote{R1 doesn't like the term "objective"} and positive sampling strategies in Section~\ref{ssec:related_positive_sampling} and show that this is an orthogonal research direction to negative sampling. \srsf{Section~\ref{ssec:contrastive_learning} positions our work viz.\ the orthogonal direction of contrastive learning.}  Finally, in Section~\ref{ssec:large_vocabulary_bottleneck}, we discuss how similar problems are solved in language models and why these solutions are not applicable to recommendations. 

\subsection{\srs{Negative Sampling Heuristics: Hard Negatives, Informative Samples, Popularity Sampling}}\label{ssec:neg_sampling}
\looseness -1 One of the first attempts to train recommender systems with negative sampling was Bayesian Personalised Rank (BPR)~\cite{rendleBPRBayesianPersonalized2009}. \sm{The authors of BPR observed that models tend to predict scores close to exactly one for positive items in the training data (\craig{a form of} overconfidence) and proposed to sample one negative item for each positive item and optimise the relative order of these items, instead of the absolute probability of each item to be positive.} However, as Rendle (the first author of BPR) has recently shown~\cite{rendleItemRecommendationImplicit2022}, BPR optimises the Area Under Curve (AUC) metric, which is not \sfdtd{top-heavy} and is therefore not \swtd{most effective} for a ranking task. Hence, several improvements over BPR, such as  WARP~\cite{westonWsabieScalingLarge2011}, \dd{LambdaRank~\cite{burgesRanknetLambdarankLambdamart2010}}, LambdaFM~\cite{yuanLambdaFMLearningOptimal2016}, and adaptive item  
sampling~\cite{rendleImprovingPairwiseLearning2014} have since been proposed to make negatively-sampled recommender models more suitable for \swtd{top-heavy} ranking tasks. \srs{These approaches usually try to mine the most informative (or \emph{hard}) negative samples that erroneously have high scores and therefore are ranked high.}
Unfortunately, these approaches mostly rely on iterative \dd{or sorting-based} sampling techniques that are not well-suited for neural network-based approaches
used by sequential recommendation models: neural models are usually trained on GPUs, which allow efficient parallelised computing, but
perform poorly with such iterative methods. Indeed, Chen et al.~\cite{chenGeneratingNegativeSamples} recently proposed an iterative sampling procedure for sequential recommendation, but only experimented with smaller datasets (<30k items) where \scrc{state-of-the-art} results can be achieved without sampling at all (see also Section~\ref{sssec:results:bert4rec_performance}). 
Instead, sequential recommenders typically rely on simple heuristics such as uniform random sampling (used by Caser~\cite{Caser} and SASRec~\cite{SASRec}) or do not use negative sampling at all (e.g.\ BERT4Rec~\cite{BERT4Rec}). \sfdtd{Pellegrini et al.~\cite{pellegriniDonRecommendObvious2022a}  recently proposed to sample negatives according to their popularity and showed \craig{this to be} beneficial when the evaluation metrics are also popularity-sampled. Our initial experiments have shown that popularity-based sampling is indeed beneficial with popularity-based \dd{evaluation} metrics, but not with the full (unsampled) metrics. However, several recent publications~\cite{kricheneSampledMetricsItem2022, canamaresTargetItemSampling2020, dallmannCaseStudySampling2021, PetrovRSS22, petrov2023rss} recommend against using sampled metrics, and therefore we avoid popularity sampling in this paper.}

\looseness -1 \sfdtd{Another heuristic that is \sdd{popular  for search tasks  is {\em in-batch} sampling~\cite[Ch.~5]{linPretrainedTransformersText2022} (e.g.\ used by GRU4Rec~\cite{GRU4Rec}). According to \cite{rendleItemRecommendationImplicit2022}, }in-batch sampling is equivalent to popularity-based negative sampling, and \sdd{hence} we avoid it for the same reason stated above.} Indeed, we focus on uniform sampling -- as used by many sequential recommender systems -- and design a solution that helps to counter the overconfidence of such models caused by uniform sampling.

\subsection{Training Objectives}\label{ssec:related_positive_sampling}
\scrc{A} \emph{training objective} is the task that the model learns to solve during the course of training. Some of the most popular alternative training objectives for sequential recommendation models include: \emph{sequence continuation}, where the model learns to predict one or several next items in the sequence (used by Caser~\cite{Caser}); \emph{sequence shifting}, where the model learns to shift the input sequence by one element to the left (used by SASRec~\cite{SASRec} and NextItNet~\cite{yuanSimpleConvolutionalGenerative2019}); item masking (used by BERT4Rec~\cite{BERT4Rec}); recency-based sampling, where the target items are selected probabilistically with a higher chance of selecting recent items (used by SASRec-RSS~\cite{PetrovRSS22, petrov2023rss}). Each of these \sfdtd{training objectives requires} negative interactions \craig{in order to} train the model to distinguish them from the positive ones. \sfdtd{\sm{Therefore,} the negative sampling strategy can be seen as orthogonal to the training objective}. 
Hence, in this paper, we only focus on negative sampling, using \scrc{the} classic SASRec \scrc{model}~\cite{SASRec}\scrc{,} with its sequence shifting training task\scrc{,} as our \scrc{"}backbone model\scrc{"}.

\subsection{\srs{Contrastive Learning}} 
\label{ssec:contrastive_learning}
\srs{
    \inote{R1: not clear why we need this section}
    \scrc{In this section, we briefly discuss contrastive learning methods, which have recently \scrci{been shown to be effective} in sequential recommendation~\cite{DuoRec, xieContrastiveLearningSequential2022, CBiT, zhouS3RecSelfSupervisedLearning2020}; the main goal of this discussion is to highlight the orthogonality of these methods to our research. Contrastive learning methods} augment the main training objective with  \scrc{an} auxiliary contrastive objective to help the}
    model to learn more generic sequence representations. The idea is to generate several versions of the same sequence (e.g. crop, reverse, add noise etc.) and add an auxiliary loss function that ensures that two versions of the same sequence have similar latent representations while representations of different sequences are located far away from each other in the latent space. This allows the model to learn more robust representations of sequences and generalise better to new sequences. However, these contrastive models still require regular training objectives and  loss functions and, therefore, also require negative sampling when the catalogue size is large. Hence, contrastive learning is an orthogonal direction, and auxiliary contrastive loss can be used with the methods described in this paper. However, in Section~\ref{ssec:sota_performance}, we demonstrate that \ourmodel{} can achieve quality comparable with the best contrastive methods even without auxiliary training objectives and loss functions.

\subsection{Large Vocabularies in Language Models}\label{ssec:large_vocabulary_bottleneck}
\looseness -1 In Natural Language Processing, the problem aligned to a large catalogue size is known as the \emph{large vocabulary bottleneck}. Indeed, according to \sm{Heap's Law~\cite{heapsInformationRetrievalComputational1978},} the number of different words in a text corpus grows with the size of the corpus, reaching hundreds of \sd{billions} of words in recent corpora~\cite{dodgeDocumentingLargeWebtext2021a}, and making computing scores over all possible words in a corpus problematic.   
A typical solution employed by modern deep learning language models is to use \emph{Word Pieces}~\cite{wuGoogleNeuralMachine2016}, which splits infrequent words \dd{into (more frequent)} sub-word groups of characters.  
This allows to use a vocabulary of relatively small size (e.g.\ $\sim$30,000 \dd{tokens} in BERT\sd{~\cite{BERT}}) whilst being capable of modelling millions of words by the contextualisation of the embedded word piece representations. \scrc{While decomposing item ids into \scrci{sub-items} can be used to reduce the item vocabulary of a recommender~\cite{petrov2023generative}, the decomposition requires a more complex two-stage learning process to assign \scrci{sub-items}}.
Other techniques have also been proposed to reduce the vocabulary size by pruning some tokens. For example, some classification models remove non-discriminating words~\cite{stolckeEntropybasedPruningBackoff1998, acquavia2023static}, which in the context of recommender systems means removing popular items (e.g.\ if a movie was watched by most of the users, it is not-discriminating). \sfdtd{However, removing popular items is a bad idea as users are prone to  \sm{interact} with popular items and recommending popular items is a strong baseline}~\cite{jiRevisitPopularityBaseline2020}. %
Perhaps the most related work to ours is the Sampled Softmax loss~\cite{jeanUsingVeryLarge2015}, which proposes a mechanism to approximate the value of a Softmax function using a small number of negatives. However, Softmax loss is known to be prone to overconfidence~\cite{weiMitigatingNeuralNetwork2022a}. Indeed, Sampled Softmax loss has recently been  shown to incorrectly estimate the magnitudes of the scores in the case of recommender systems~\cite{wuEffectivenessSampledSoftmax2022}. \srs{Our experiments with Sampled Softmax loss are aligned with these findings. We discuss Sampled Softmax loss in detail in Section~\ref{ssec:recsys:loss} and experimentally evaluate it in Section~\ref{ssec:gbce_effect}}.
In summary, among the related work, there is no solution to the overconfidence problem %
in sequential recommender systems. Hence, we aim to close this gap and design a solution for this overconfidence that is suitable for sequential models. In the next section, we cover the necessary required preliminaries and then in Section~\ref{sec:gbce_and_properties}, we show that the problem can be solved with the help of Generalised Binary Cross-Entropy loss.

\section{Sequential Recommendation \& loss Functions}\label{sec:background}
In the following, Section~\ref{ssec:backgournd:transformers} describes the SASRec and BERT4Rec sequential recommendation models, which form the backbone of this paper.
In Section~\ref{ssec:background:sequential_recsys}, we more formally set the sequential recommendation task as a probabilistic problem and in Section~\ref{ssec:recsys:loss} discuss loss functions used for training sequential models.  
\subsection{SASRec and BERT4Rec}\label{ssec:backgournd:transformers}
\swtd{Transformer~\cite{Transformer}-based models have recently outperformed other models in Sequential Recommendation~\cite{SASRec, BERT4Rec,DuoRec,Bert4RecRepro, PetrovRSS22, petrov2023rss}}.
Two of the most popular Transformer-based recommender models are BERT4rec~\cite{BERT4Rec} and  SASRec~\cite{SASRec}. 
 \sm{The key differences between the models} include different attention mechanism (bi-directional vs.\ \sdd{uni-directional}), different training objective (Item Masking vs.\ Shifted Sequence), different loss functions (Softmax loss vs.\ BCE loss), and, importantly, \sm{different negative sampling strategies} (BERT4Rec does not use sampling, whereas SASRec samples 1 negative per positive).  

BERT4Rec was published one year later compared to SASRec, and in the original publication~\cite{BERT4Rec}, 
Sun et al.\ demonstrated \dd{the} superiority of BERT4Rec over SASRec. Petrov and Macdonald confirmed this superiority
in a recent \sd{replicability} study~\cite{Bert4RecRepro},
\craig{and observed that, when fully-converged, BERT4Rec still exhibits state-of-the-art performance, 
outperforming many later models.}

\looseness -1 Sun et al.~\cite{BERT4Rec} attributed BERT4Rec's high effectiveness to its bi-directional attention mechanism.
\sm{Contrary to that}, our theoretical analysis and experiments show that it should be attributed to \scny{the model overconfidence caused by the negative sampling} used by SASRec (see Section~\ref{sssec:results:bert4rec_performance}). Indeed, when controlled \sm{for negative sampling}, these models perform similarly (e.g. SASRec also exhibits state-of-the-art performance when trained without negative sampling)\inote{R2 wants us to mention that this was also shown in the RSS paper (not exactly) and some random workshop paper (not sure if worth mentioning)}. 
\srs{Unfortunately, as we argue in Section~\ref{sec:intro}, the large size of the item catalogue in many real-world systems means that} \scrc{using} negative sampling in \scrc{the training of} such systems is unavoidable, \srsd{and therefore these systems can not use models that do not use sampling, such as BERT4Rec.} Our goal hence is to improve SASRec's performance \dd{(by addressing overconfidence) while retaining the negative sampling, which is needed for large-scale systems.} 

We now discuss a probabilistic view of sequential recommendation, which we use for improving SASRec in Section~\ref{sec:gbce_and_properties}.

\subsection{Probabilistic View of Sequential Recommendation}\label{ssec:background:sequential_recsys}
\looseness -1 The goal of a sequential recommender system is to predict \sm{the next item} in a sequence of user-item interactions. 
Formally, given a sequence of user-item interactions $u=\{i_0, i_1, i_2, ... i_n\}$, where $i_k \in I$, the goal \swtd{of the model is to predict the next user's interaction $i_{n+1}$}. 
Sequential recommendation is usually cast as a \emph{ranking problem}, so \emph{predict} means to rank items in the catalogue according to their estimated probability of appearing next in the sequence. \craig{We denote this (\emph{prior}) probability distribution over all items appearing next in the sequence after $u$ as $P(i|u)$. $P(i|u)$ is not directly observable: \sm{the training data only contains the user's actual interactions and does not contain information about the probabilities of any alternative items not interacted with.} We refer to the prior as $P(i)$ for simplicity.}

\craig{Learning to estimate the prior distribution $P(i)$ is a hard task because the model doesn't have access to it, even during training. Instead, the model learns to estimate these probabilities, i.e.\ $\hat{p}$ = $\{\hat{p}_1, \hat{p}_2, ...,  \hat{p}_{|I|}\}$, by
using a posterior distribution $y(i) = \I[i=i^+]$,} where $i^+ \in I$ is a positive interaction selected according to the training objective (as discussed in Section~\ref{ssec:related_positive_sampling}). $y(i)$ is measured \emph{after} the user selected the item, so it always equals 1 for the positive item $i^+$ and equals 0 for all other items.

Note that to rank items, models do not have to compute \dd{the} modelled probabilities $\hat{p}$ explicitly. 
Instead, models frequently compute item \emph{scores} $s = \{s_1, s_2, ...,  s_{|I|}\}$ and assume that if item $i$ is scored higher than item $j$ ($s_i > s_j$) then 
item $i$ is more likely to appear next in the sequence than item $j$ ($\hat{p_i} > \hat{p_j}$). Whether or not it is possible to recover modelled item probabilities $\hat{p} = \{\hat{p}_1, \hat{p}_2, .., \hat{p}_{|I|}\}$ from the scores $s$ depends on the loss function used for model training. 

We say that a loss function $\mathcal{L}$ \emph{directly models probabilities $\hat{p}$}, if there exists a function $f$, \sd{which converts scores to probabilities} ($\hat{p_i} = f(s_i)$) and when the model is trained with $\mathcal{L}$, $\hat{p}$
approximates the prior distribution $P$ (e.g. a model trained with $\mathcal{L}$ minimises the KL divergence between $P$ and $\hat{p}$). 
\sfdtd{In the next section, we discuss the loss functions used by sequential models that directly model probabilities.}

\subsection{BCE \sdd{Loss} and Softmax \sdd{Loss}}\label{ssec:recsys:loss}
Two popular loss functions, which directly model probabilities are \emph{Binary Cross-Entropy (BCE)} (used by Caser~\cite{Caser} and SASRec~\cite{SASRec})
and {\em Softmax \sdd{loss}} (used by BERT4Rec~\cite{BERT4Rec} and ALBERT4Rec~\cite{Bert4RecRepro}).

\looseness -1 Binary Cross-Entropy is a \emph{pointwise} loss, which treats the ranking problem as a set of independent binary classification problems. It models the probability \sd{with the help of the {\em logistic sigmoid function $\sigma(s)$}:}
\begin{align}
    \hat{p_i} = \sigma(s_i) =  \frac{1}{1+e^{-s_i}} \label{eq:porb_is_sigmoid}
\end{align}
The value of BCE loss is then computed as: 
\begin{align}
    \LossBCE = -\frac{1}{|I|} \sum_{i \in I} y(i)\log(\hat{p_i}) + (1-y(i))\log(1-\hat{p_i})
\end{align}

\looseness -1 BCE minimises \scrc{the} KL divergence~\cite[Ch.~5]{murphyProbabilisticMachineLearning2022} between \scrc{the} posterior and \scrc{the} modelled \scrc{distributions,} $D_{KL}(y(i)||p_i)$, where each \scrc{of the probability distributions} is treated as a distribution with two outcomes \sfdtd{(i.e.\ interaction/no interaction)}. BCE \dd{considers} each probability independently, \srsd{so} their sum does not have to add up to 1.
Indeed, as we show in Section~\ref{sec:gbce_and_properties}, when BCE is used with negative sampling, the model learns to predict probabilities close to 1 for the most highly-ranked items. 

\looseness -1 In contrast, Softmax loss treats the ranking problem as a multi-class classification problem, \dd{thereby considering}
the probability distribution across all items, obtained by using \dd{a} $\text{softmax}(\cdot)$ operation:
\begin{align}
    \hat{p_i} = \text{softmax}(s_i) =  \frac{e^{s_i}}{\sum_{j \in I}e^{s_j}}\label{eq:softmax}
\end{align}
The value of Softmax loss is then computed as: 
\begin{align}
    \mathcal{L}_{softmax} = -\sum_{i \in I} y(i)\log(\hat{p}_i) = -\log(\text{softmax}(s_{i^+}))
\end{align}
Softmax loss minimises KL divergence~\cite[Ch.~5]{murphyProbabilisticMachineLearning2022} between posterior and modelled distributions $D_{KL}(y||p)$, where each $y$ and $p$ are multi-class probability distributions.  In contrast to BCE, \sm{the item probabilities} $\hat{p}_{i}$ modelled by Softmax loss add up to 1, \sm{meaning that overconfidence is less prevalent} (however, it is still known to overestimate probabilities \sdd{of} the top-ranked items~\cite{weiMitigatingNeuralNetwork2022a}). 

Unfortunately, the $\text{softmax}(\cdot)$ operation used by Softmax loss \sfdtd{requires access to \emph{all} item scores to compute the probabilities \sfdtd{(which makes it more of a \emph{listwise} loss)}, whereas if the model is trained with negative sampling, the scores are only computed for the \emph{sampled} items. This makes Softmax loss unsuitable for training with negative sampling. In particular, this means that BERT4Rec, which uses the Softmax loss, cannot be trained with sampled negatives (without changing the loss function).} 

\looseness -1 \srs{To use Softmax loss with sampled negatives, Jean et\ al.~\cite{jeanUsingVeryLarge2015} proposed \emph{Sampled Softmax Loss (SSM)}. SSM approximates probability  $\hat{p_i}$ from Equation~\eqref{eq:softmax} using a subset of $k$ negatives $I^-_k \subset I^-$. This approximation is then used to derive the loss:}
\srs{
    \begin{align}
            \hat{p_i} = \text{SSM}(s_i, I^-_k) =  \frac{e^{s_i}}{e^{s_{i^+}} + \sum_{j \in I^-_k}e^{s_j}}\label{eq:ssm}\\
            \mathcal{L}_{SSM} = -\sum_{i \in \{I_k^- \cup {i^+}\}} y(i)\log(\hat{p}_i) = -\log(\text{SSM}(s_{i^+}))
    \end{align}
}
\srs{
    The estimated probability value computed with Sampled Softmax is higher than the probability estimated using full Softmax, as the denominator in Equation~\eqref{eq:softmax} is \srsf{larger} than the denominator in Equation~\eqref{eq:ssm}. However, if all high-scored items are included in the sample $I^-_k$, the approximation becomes close. To achieve this, Jean et\ al. originally proposed a heuristic approach specific to textual data (they segmented texts into chunks of related text, where each chunk had only a limited vocabulary size). In the context of sequential recommender systems, some prior works~\cite{yuanSimpleConvolutionalGenerative2019, pellegriniDonRecommendObvious2022a}  used variations of SSM loss with more straightforward sampling strategies, such as popularity-based or uniform sampling. 
    In this paper, we focus on the simplest scenario of uniform sampling, and therefore in our experiments, we use Sampled Softmax loss with uniform sampling.    
    Note that Sampled Softmax Loss normalises probabilities differently compared to the full Softmax loss, and therefore the Sampled Softmax loss and the full Softmax loss are different loss functions. Indeed, as Sampled Softmax uses only a sample of items in the denominator of Equation~\eqref{eq:ssm}, the estimated probability of the positive item $\hat{p_i}$ is an overestimation of the actual probability, a form of overconfidence. \srsd{Indeed, as mentioned above,  Sampled Softmax loss fails to estimate probabilities accurately for recommender systems~\cite{wuEffectivenessSampledSoftmax2022}}. Nevertheless, as variations of Sampled Softmax have been used in sequential recommendations~\cite{yuanSimpleConvolutionalGenerative2019, pellegriniDonRecommendObvious2022a}, we use Sampled Softmax loss as a baseline in our experiments (see Section~\ref{ssec:gbce_effect}).
}

\sm{In contrast, it is possible to calculate BCE loss over a set of sampled negatives $I^-_k$ \srs{without modifying the loss itself (except for a normalisation constant, which does not depend on the item score and therefore can be omitted)}, as follows:} 
    \begin{align}
        \LossBCE =-\frac{1}{|I_k^-| + 1} \left( \log(\sigma(s_{i^+})) + \sum_{i \in I_k^{-}}\log(1-\sigma(s_i)) \right) \label{eq:bce_sampled}
    \end{align}
\sfdtd{Using BCE with sampled negatives is a popular approach, applied by  models such as SASRec~\cite{SASRec} (which uses 1 negative per positive), \dd{and} Caser~\cite{Caser} (which uses 3 negatives). 
Unfortunately, negative sampling used with Binary Cross-Entropy leads to model overconfidence, which we discuss in the next section.}     

\section{Model Overconfidence}\label{sssec:overconfidence}
\begin{figure}
    \centering
    \includegraphics[width=\linewidth]{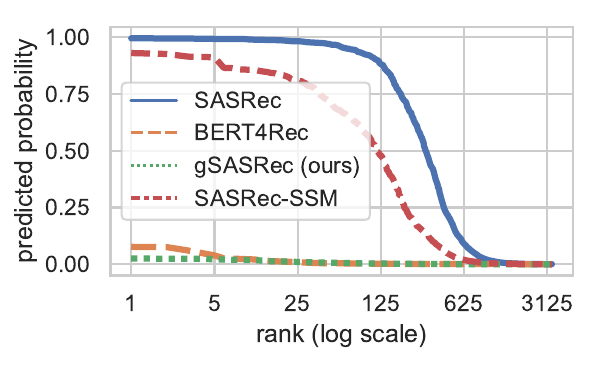}
    \caption{Predicted probability at different ranks for user 963 in MovieLens-1M. \srsf{SASRec-SSM is a SASRec model trained with Sampled Softmax loss with 16 negatives.}}\label{fig:rank_prob_label}
\end{figure}

\looseness -1 We say that a model is {\em overconfident} in its predictions if its predicted probabilities $\hat{p}_i$ for \sfdtd{highly-scored items} are much larger compared to prior probabilities $P(i)$, i.e., $ \hat{p}_i \gg P(i)$.
\craig{In general, the magnitude of the relevance estimates are rank-invariant, i.e.\ do not affect the ordering of items, and hence they are rarely considered important when formulating a ranking model. \dd{In contrast, overconfidence is problematic} only for the loss functions used to train the models, particularly when they directly model the interaction probability. Indeed,} for some loss functions (such as pairwise BPR~\cite{rendleBPRBayesianPersonalized2009} or listwise LambdaRank~\cite{burgesRanknetLambdarankLambdamart2010}), only the difference between the scores of paired items ($s_i - s_j$) is important, and therefore we cannot define overconfidence for these losses. However, these losses usually require algorithms that iteratively select ``informative" negative samples, which are hard to apply with deep learning methods (see also Section~\ref{ssec:neg_sampling}).
As discussed in Section~\ref{ssec:background:sequential_recsys}, the prior probability distribution $P(i)$ cannot be directly observed, and therefore overconfidence may be hard to detect. However, in some cases, overconfidence may be obvious.  For example, Figure~\ref{fig:rank_prob_label} shows predicted probabilities by \srsf{four} different models for a sample user in the MovieLens-1M dataset. As can be seen from the figure, SASRec's predicted probabilities for items at positions 1..25 are almost indistinguishable from 1. This is a clear sign of overconfidence: only one of these items can be the correct prediction, and therefore we expect the \emph{sum} of probabilities to be approximately equal to 1, and not each individual probability. In fact, in this  figure, the sum of all probabilities predicted by SASRec equals 338.03. In contrast, for BERT4Rec, the sum of probabilities equals exactly 1 (as \dd{the probabilities are} computed using Softmax) and for our \ourmodel{}~(see Section~\ref{ssec:gsasrec})~it is equal to 1.06. \srsd{From the figure, we also see that a SASRec model trained with Sampled Softmax loss is also prone to overconfidence (sum of all probabilities equals 152.3).}

\looseness -1 Overconfidence \dd{for highly-ranked} items is problematic: the model does not learn to distinguish these items from each other (all their predicted probabilities are approximately equal) and focuses on distinguishing top items from the bottom ones. \scrc{The lack of focus on the top items} contradicts our goal: we want \dd{the correct order of highly-ranked items and are not interested in} score differences beyond a certain cutoff. Moreover, overconfidence is specifically problematic \scrc{for BCE loss}: if an item with high probability  $\hat{p}_i \approx 1$ is selected as a negative (the chances of such an event are high when there are many high-scored items), $\log(1-{p}_i)$ computed by the loss \scrc{function} tends to $-\infty$, causing numerical overflow problems and \scrci{a} training instability. 

\looseness -1 \shd{Next, we introduce gBCE loss, apply it for a theoretical analysis of BCE's overconfidence, and show how gBCE mitigates the overconfidence problem.}

\section{Generalised Binary Cross Entropy and its properties}\label{sec:gbce_and_properties}  
In this section we design \ourloss{}~and theoretically show that it can mitigate the overconfidence problem. 
In Section~\ref{ssec:gbce} we introduce \ourloss{}~and analyse its properties; in Section~\ref{ssec:gbce_bce_relation} we show that \ourloss{}~may be replaced with regular BCE loss with transformed positive scores, which may be more convenient in practice; in Section~\ref{ssec:t} we show how to reparametrise~\ourloss{}~to make it independent from the chosen sampling rate; finally, in Section~\ref{ssec:gsasrec} we introduce \ourmodel{} -- an improved version of SASRec, which uses \ourloss{}.

\subsection{Generalised Binary Cross Entropy}\label{ssec:gbce}
We now introduce \emph{Generalised Binary Cross Entropy (\ourloss{})} loss,
which we use to analyse and mitigate overconfidence induced by negative sampling. We define \ourloss{}, parameterised \craig{by $\beta$} as:
    \begin{align}
        \LossGBCE =-\frac{1}{|I_k^-| + 1} \left( \log(\sigma^{\beta}(s_{i^+})) + \sum_{i \in I_k^{-}}\log(1-\sigma(s_i)) \right). \label{eq:gbce}
    \end{align}
\ourloss{}~differs from regular BCE loss in that it uses the \stdtd{\emph{generalised logistic sigmoid function}~\cite{richardsFlexibleGrowthFunction1959, phamCombinationAnalyticSignal2019} for the positive sample} (sigmoid raised to the power of $\beta$). \stdtd{The power parameter \dd{$\beta\geq0$ controls the shape of the generalised sigmoid. For example, when $\beta \approx 0$, the output of the generalised sigmoid becomes closer to 1 for all input scores.} On the other hand,}
when $\beta=1$, \stdtd{BCE and \ourloss{}} are equal:
\begin{align}
    \LossGBCEb{1}=\LossBCE \label{eq:gbce_equal_bce}
\end{align}

Similarly to BCE, \ourloss{}~is also a pointwise loss, and \sdd{it considers} the probability of interaction as \sm{a sigmoid transformation} of the model score (Equation~\eqref{eq:porb_is_sigmoid}). We now show the exact form of the relation between the prior probability $P(i)$ (which we desire to estimate, as we discuss in Section~\ref{ssec:background:sequential_recsys}) and the modelled probabilities $\hat{p}_i = \sigma(s_i)$, learned by a model trained with~\ourloss{}. 
 
\begin{theorem}
    \sdd{For every user in the dataset}, let $P(i)$ be the prior probability distribution of \sdd{the} user interacting with item $i \in I$, 
    \sfdtd{$s = \{s_1, ... s_{|I|}\}$ are scores} predicted by the model, $i^{+}$ is a positive sample selected by the user,
    $I_k^{-} = \{i^{-}_1, i^{-}_2, ...,  i^{-}_k\}$ - k randomly (uniformly, with replacement) sampled 
    negatives, $\alpha = \frac{k}{|I^-|}$ - negative sampling rate. 
    Then a recommender model\scrc{,} trained on a sufficiently large number of training 
    samples using gradient descent and $\LossGBCE$ loss\scrc{,} will converge to predict score distribution $s$, so that \label{theorem:bias}
    \begin{align}
        \sigma(s_i) = \frac{\beta P(i)}{\alpha - \alpha P(i) + \beta P(i)}; 
        \forall i \in I \label{eq:theorem}
    \end{align}
\end{theorem}
\begin{proof}
    With a sufficiently large number of training samples, gradient descent converges to minimise the expectation of the loss \sm{function~\cite[Ch.~4]{goodfellowDeepLearning2016} \craig{(assuming the expectation has no local minima)}}. 
    Therefore, the predicted score distribution converges to the minimum of \sfdtd{the expectation $\E\left[\LossGBCE\right]$}:
    \begin{align}
        s = \argmin_{s}\E\left[\LossGBCE\right]
    \end{align}
    Hence, our goal is to show that Theorem~\ref{theorem:bias} is true if and only if the expectation $\E\left[\LossGBCE\right]$ is minimised. 
    
    To show that, we first rewrite the definition of $\LossGBCE$~(Equation~\eqref{eq:gbce}) as a sum of contributions \craig{for each
    individual item in $I$}: 
    \begin{align}
        \LossGBCE = \frac{1}{|I^-_k|+1}\sum_{i \in I} {\Loss}_i
    \end{align}
    where the contribution of each item, $\mathcal{L}_i$, is defined as follows: 
    \begin{align}
        {\Loss}_i = -(\I[i = i^+] \log(\sigma^\beta(s_{i})) + \sum_{j=1}^{k} \I[i=i^-_j]\log(1-\sigma{(s_{i})}))
   \end{align}
   The probability of an item being selected as a positive is defined by the prior distribution:
   \begin{align}
   P(\I[i=i^+]) = P(i)\label{eq:positive_prob}
   \end{align}
   whereas the probability of an item being selected as $j^{th}$ negative is equal to the product of the probability of an item being negative and the negative sampling probability. If we apply a uniform sampling with a replacement for identifying negatives, then the sampling probability is always equal to $\frac{1}{|I^-|}$, so  overall, the probability of selecting an item $i$ as the $j^{th}$ negative can be written as:
   \begin{align}P(\I[i=i^-_j]) = \frac{1}{|I^-|}(1-P(i))\label{eq:negative_prob}\end{align}
   We can now calculate the expectations of each individual loss contribution $\E[{\Loss}_i]$:  
    \begin{align}
     \E[{\Loss}_i] &= -(P(\I[i = i^+]) \log(\sigma^\beta(s_{i})) + \sum_{j=1}^{k} P(\I[i=i^-_j])\log(1-\sigma{(s_{i})}))  \nonumber\\& \text{\;\;\;\;\;\;\;\;\;(By the definition of expectation)} \nonumber\\
        &= -(P(i) \log(\sigma^\beta(s_{i})) + \sum_{j=1}^{k} \frac{1}{|I^-|}(1-P(i))\log(1-\sigma{(s_{i})})) \nonumber\\& \text{\;\;\;\;\;\;\;\;\;(Substituting Equations~\eqref{eq:positive_prob}~and~\eqref{eq:negative_prob})} \nonumber\\
        &= -(P(i) \log(\sigma^\beta(s_{i})) +\frac{k}{|I^-|}(1-P(i))\log(1-\sigma{(s_{i})}))  \nonumber\\&  \text{\;\;\;\;\;\;\;\;\;(The sum is just the same term repeated $k$ times)} \nonumber\\
        &= -(P(i) \log(\sigma^\beta(s_{i})) +\alpha(1-P(i))\log(1-\sigma{(s_{i})})) \nonumber\\& \text{\;\;\;\;\;\;\;\;\;(Substituting the sampling rate definition $\alpha = \frac{k}{|I^-|}$)}
        \label{eq:loss_components_expectation}
    \end{align}
      Differentiating Equation~\eqref{eq:loss_components_expectation} on $\sigma(s_i)$ we get: 
    \begin{align}
        \frac{d \E[{\Loss}_i]}{d \sigma(s_i)} = -\frac{\beta P(i)}{\sigma(s_i)} + \frac{\alpha(1-P(i))}{1-\sigma(s_i)} \label{eq:loss_derivative}
    \end{align}
    \sfdtd{Our goal is to minimise the expectation $ \E[{\Loss}_i]$, so equating this derivative to zero and solving for 
    $\sigma(s_{i})$ we obtain the value of $\sigma(s_{i})$, which minimises the expectation}:
    \begin{align}
        \sigma(s_i) = \frac{\beta P(i)}{\alpha - \alpha P(i) + \beta P(i)}
        \label{eq:desired_equating}
    \end{align} 
    
   We now rewrite the expectation $\E\left[\LossGBCE\right]$  as the sum of its individual components: 
   \begin{align}
       \E\left[\LossGBCE\right] = \E\left[\frac{1}{|I^-_k| + 1}\sum_{i \in I}{\Loss}_i\right] = \frac{1}{|I^-_k|+1}\sum_{i \in I}\E\left[{\Loss}_i\right]
    \label{eq:exp_sum_contributions}
   \end{align}

  According to Equation~\eqref{eq:exp_sum_contributions}, the expectation $\E\left[\LossGBCE\right]$ is minimised when, for each $i \in I$,
   the individual contributions $\E[{\Loss}_i]$ are minimised, i.e.\ when Equation~\eqref{eq:desired_equating} is true for each $i \in I$. 
\end{proof}

We now use Theorem~\ref{theorem:bias} to analyse properties of both regular and generalised Binary Cross-Entropy losses.
First, we show that it is possible to train a model to estimate a prior distribution $P(i)$ exactly using \ourloss{}~loss. 
\begin{corollary}
    If a model is trained using negative sampling with sampling rate $\alpha \le 1$ and \ourloss{}~loss $\LossGBCE$ with $\beta=\alpha$,
    \sfdtd{then the model converges to predict probabilities calibrated with the prior distribution}:\label{col:generalised_unbiased}
    \begin{align}
        \sigma(s_i) = P(i) \label{eq:gen_unbiased}
    \end{align} 
\end{corollary}
\begin{proof}
    We can obtain Equation~\eqref{eq:gen_unbiased} by substituting $\beta=\alpha$ in Equation~\eqref{eq:theorem}.
\end{proof}
We now use Theorem~\ref{theorem:bias} to analyse properties of regular Binary Cross-Entropy \sd{loss}. 
\begin{corollary}
    If a model is trained with BCE loss $\LossBCE$ 
    and negative sampling, with  sampling rate $\alpha$, then it converges to predict \sfdtd{scores $s_i$ so that}\label{col:bce_bias}
    \begin{align}
        \sigma(s_i) = \frac{P(i)}{\alpha - \alpha P(i) + P(i)}\label{eq:bias_bce}
    \end{align}\label{col:bce_bias}
\end{corollary}
\begin{proof}
    According to Equation~\eqref{eq:gbce_equal_bce}, $\LossBCE$ is equal to $\LossGBCE$ with $\beta=1$.
    Substituting $\beta=1$ into Equation~\eqref{eq:theorem} we obtain Equation~\eqref{eq:bias_bce}.
\end{proof}
We can now show that SASRec learns an overconfident \craig{score} distribution:
\begin{corollary}
    \srsd{The SASRec model with $\LossBCE$ and one negative per positive converges to yield scores $s_i$, such that}:
    \label{corollary:sasrec}
    \begin{align}
        \sigma(s_i) = \frac{P(i)|I| - P(i)}{P(i)|I| - 2P(i) + 1} \label{eq:sasrec_bias}
    \end{align}
    \begin{proof}
        SASRec uses one negative per positive, meaning that its sampling rate is equal to: 
        \begin{align}
            \alpha = \frac{1}{|I| - 1} \label{eq:sasrec_sampling_rate}
        \end{align}
        Substituting Equation~\eqref{eq:sasrec_sampling_rate} into Equation~\eqref{eq:bias_bce}, we get Equation~\eqref{eq:sasrec_bias}.
    \end{proof}
\end{corollary}
Corollary~\ref{corollary:sasrec} explains why SASRec tends to predict very high probabilities for top-ranked items: when an item has a higher-than-average probability of being selected ($P(i) \gg \frac{1}{|I|}$), the term $P(i)|I|$ dominates both the numerator and denominator of Equation~\eqref{eq:sasrec_bias}, meaning that the predicted probability $\sigma(s_i)$ will be very to close to~1. 
\subsection{Relation between BCE and \ourloss{}}\label{ssec:gbce_bce_relation}
\looseness -1 In Section~\ref{ssec:gbce} we showed that \ourloss{}~is equal to regular BCE loss when the \stdtd{power} parameter $\beta$ is set to 1. 
We now show that these two loss functions have a deeper relation, which allows using well-optimised versions of BCE from deep learning frameworks instead of~\ourloss{}. \label{theorem:gbce_bce_equality}%
\begin{theorem}
    Let  $s^{+}$ be the predicted score for a positive item and $s^- = \{s_{i^-_1}, s_{i^-_2}..s_{i^-_{|I^{-}|}}\}$ be the predicted scores for the negative items. 
    Then 
    \begin{align}
        \LossGBCE(s^+, s^-) = \LossBCE(\gamma(s^+), s^-)
    \end{align}
where 
\begin{align}
    \gamma(s^+)= \log\left(\frac{1}{\sigma^{-\beta}(s^+) - 1}\right) \label{eq:scores_transformation}
\end{align}

\end{theorem}%
\begin{proof}
   
    According to the definition of the logistic sigmoid function (Equation~\eqref{eq:porb_is_sigmoid}), 
    \begin{align}
        \sigma(\gamma(s^+)) &= \frac{1}{e^{-\gamma(s^+)}+1} \nonumber \\
        &= \frac{1}{e^{-\log\left(\frac{1}{\sigma^{-\beta}(s^+) - 1}\right)}+1} \nonumber\\& \text{\;\;\;\;\;\;\;\;\;(Substituting $-\gamma(s^+)$ with its definition (Eq.~\eqref{eq:scores_transformation}))} \nonumber\\
        &= \frac{1}{e^{\log\left(\sigma^{-\beta}(s^+) - 1\right)}+1} \nonumber\\& \text{\;\;\;\;\;\;\;\;\;(Using properties of the $\log(\cdot)$ function)} \nonumber \\
        &= \frac{1}{\sigma^{-\beta}(s^+) - 1 + 1} \nonumber\\& \text{\;\;\;\;\;\;\;\;\;(The exponent and the logarithm cancel each other out)} \nonumber \\
        & = \frac{1}{\sigma^{-\beta}(s^+)}  = \sigma^{\beta}(s^+)
    \end{align}
       
    Substituting $\sigma^{\beta}(s^+) = \sigma(\gamma(s^+))$ into the definition of $\LossGBCE$ (Equation~\eqref{eq:gbce}) 
    and taking into account the definition of $\LossBCE$ (Equation~\eqref{eq:bce_sampled}) we get the desired equality:
    \begin{align}
     \LossGBCE(s^+, s^-) &=-\frac{1}{|I^-_k| + 1} \left( \log(\sigma^{\beta}(s_{i^+})) + \sum_{i \in I^{-}}\log(1-\sigma(s_i)) \right) \nonumber \\  
      &= -\frac{1}{|I^-_k| + 1} \left( \log(\sigma(\gamma(s^+))) + \sum_{i \in I^{-}}\log(1-\sigma(s_i)) \right) \nonumber \\
      &= \LossBCE(\gamma(s^+),s^-) \nonumber
    \end{align}
\end{proof}

In practice, Theorem~\ref{theorem:gbce_bce_equality} allows us to transform the predicted positive scores by using Equation~\eqref{eq:scores_transformation} and then train the model using the regular BCE loss, instead of using \ourloss{}~directly. 
This \dd{is actually} preferable because many machine learning frameworks have efficient and numerically stable implementations for standard loss functions such as BCE loss. Indeed, in our implementation, we also rely on Equation~\eqref{eq:scores_transformation} score transformation and regular BCE loss instead of using \ourloss{}~directly.

\subsection{\stdtd{Calibration Parameter} $t$} \label{ssec:t}
\looseness -1 As shown in Section~\ref{ssec:gbce}, setting \stdtd{the power parameter} $\beta=1$ in \ourloss{}~resembles the regular BCE loss, whereas setting $\beta$ equal to the sampling rate $\alpha$ results in learning \stdtd{a fully calibrated distribution}. This means that reasonable values of the $\beta$ parameter lie in the interval $[\alpha .. 1]$. In practice, we found working with this interval inconvenient: we usually do not control the $\alpha$ parameter directly and instead infer it from the number of negatives and size of the dataset. Similarly, the possible values \sdd{of} $\beta$ depend on these variables as well. To make the interval of possible values independent from $\alpha$, we control the \stdtd{power parameter $\beta$ indirectly with the help of a \emph{calibration parameter $t$}, which \craig{adjusts} $\beta$}  as follows:
\begin{align}
    \beta = \alpha \left(t\left(1 - \frac{1}{\alpha}\right) + \frac{1}{\alpha}\right)\label{eq:t}
\end{align}
This \stdtd{substitution} makes model configuration simpler: we select $t$ in the interval $[0..1]$, where
$t=0$ ($\beta=1$) corresponds to regular BCE loss, and $t=1$ ($\beta=\alpha$) corresponds to the fully \stdtd{calibrated version of \ourloss{}, which drives the model to estimate prior $P(i)$ exactly (according to Corollary~\ref{col:generalised_unbiased}).} %

\subsection{\ourmodel{}}\label{ssec:gsasrec}
\emph{\ourmodel{}}~(generalised SASRec) is a version of the SASRec model with an increased number of negatives, trained with \ourloss{}~loss. 
Compared with SASRec, \ourmodel{}~has two extra hyperparameters: (i) number of negative samples per positive $k \in [1..|I^-|]$, and (ii) parameter $t \in [0..1]$, which indirectly controls the power parameter $\beta$ 
in \ourloss{}~using to Equation~\eqref{eq:t}.
In particular, when $k=1$ and $t=0$, \ourmodel{}~\stdtd{is} the original SASRec model, as SASRec uses 1 negative per positive and \ourloss{} becomes BCE when $t=0$. \rsm{While our primary focus is on the SASRec model, it is possible to apply \ourloss{} with other models; as an example, we use it also with 
 BERT4Rec (see Section~\ref{ssec:gbce_effect}).}

In the next section we empirically evaluate~\ourmodel{}~and show that its generalisations over SASRec are indeed beneficial and allow it to match BERT4Rec's performance, while retaining negative sampling. 
\section{Experiments}\label{sec:experiments}
We design our experiments to answer the following research questions:
\label{ssec:rqs}
\begin{enumerate}[font={\bfseries}, label={RQ\arabic*}, wide, labelwidth=!, labelindent=0pt]
    \item{How does negative sampling affect BERT4Rec's performance gains over SASRec?}
    \item{What is the effect of \ourloss{}~on predicted item probabilities?}
    \label{rq:gbce_effect}
    \item {What is the effect of negative sampling rate and parameter $t$ on the performance of \ourmodel{}?}
    \label{rq:alpha_t_performance}
    \item{How does gBCE loss affect the performance of SASRec and BERT4Rec models trained with negative sampling?}
    \item How does gSASRec perform in comparison to state-of-the-art \sd{sequential} recommendation models?
    \label{rq:outperforms_sota}
    \label{rq:popularity}
\end{enumerate}

\subsection{Experimental Setup}
\subsubsection{Datasets}
\looseness -1 We experiment with three datasets: MovieLens-1M~\cite{harperMovieLensDatasetsHistory2015a}, Steam~\cite{pathakGeneratingPersonalizingBundle2017} and Gowalla~\cite{Gowalla}.
There are known limitations with MovieLens-1M~\cite{PetrovRSS22, petrov2023rss}: it is a movies ratings dataset, and users \scrc{may not} rate items in the same order as they watch them, so the task, in this case, may be described as recommending movies to rate (and not to watch) next. However, it remains one of the most popular \sdd{benchmarks} for evaluating \dd{sequential} recommender systems~\cite{SASRec, BERT4Rec, Bert4RecRepro, DuoRec, CBiT}, and more \sdd{importantly,} researchers use it consistently without additional preprocessing (the dataset is already preprocessed by its authors). %
This consistency allows us to compare results reported between different papers, and therefore we find experimenting with this dataset important. 
To stay consistent with previous research~\cite{BERT4Rec, Bert4RecRepro}, we use preprocessed versions of the MovieLens-1M and Steam datasets provided in the BERT4Rec repository\footnote{\url{https://github.com/FeiSun/BERT4Rec/tree/master/data}} and do not apply any additional preprocessing. These datasets have relatively small \dd{numbers} of items and therefore are suitable for training unsampled models such as BERT4Rec. 

As a demonstration that \ourmodel{}~is suitable for larger datasets, we also use \sdd{the} Gowalla dataset, which is known to be problematic for BERT4Rec~\cite{PetrovRSS22,petrov2023rss}. \dd{For this dataset, and following common practice~\cite{SASRec, Caser, FPMC, BERT4Rec, PetrovRSS22,petrov2023rss}, we remove users with less than 5 interactions}. Table~\ref{tb:datasets} lists salient the characteristics of all three datasets.   
We split data using \sdd{the} standard \emph{leave-one-out} approach, where we leave the last interaction for each user in the test dataset. Additionally, \sdd{for} each dataset, we randomly selected 512 users - for \sdd{these} users, we select their second last interaction and include them into a validation dataset, which we use for \dd{hyperparameter} tuning as well as to control model early stopping\footnote{All code for this paper is available at \href{https://github.com/asash/gsasrec}{https://github.com/asash/gsasrec}}. 

\begin{table}
\caption{Experimental Datasets.}\label{tb:datasets}
\centering
\resizebox{\linewidth}{!}{
\begin{tabular}{lrrr}
\toprule
Dataset& Users &  Items &  Interactions \\
\midrule
MovieLens-1M&       6,040 &       3,416 &            999,611 \\
Steam &     281,428 &      13,044 &           3,488,885 \\  
Gowalla  &      86,168 &    1,271,638 &           6,397,903 \\
\bottomrule
\end{tabular}
}
\end{table}
\subsubsection{Metrics}
\looseness -1 Until recently, a somewhat common approach in evaluating
recommender systems \rsm{on sampled metrics using} only small number of items, 
but it has been shown that this leads to incorrect evaluation results in general~\cite{kricheneSampledMetricsItem2022, canamaresTargetItemSampling2020}, 
and specifically for sequential recommender systems~\cite{dallmannCaseStudySampling2021, Bert4RecRepro}. Hence, we always evaluate all item scores at the inference stage. 
Following~\cite{Bert4RecRepro}, we evaluate our models using \sdd{the} popular Recall and NDCG metrics measured at cutoff 10. \dd{We also calculate} Recall at cutoff 1, because according to Equation~\eqref{eq:sasrec_bias}, we expect SASRec to be more overconfident on the highest-ranked metrics, and mitigating overconfidence should have a bigger effect on metrics measured at the highest cutoff. 

\subsubsection{Models}
\looseness -1 In our experiments, we compare \ourmodel{}~with the regular SASRec model, which serves as the backbone of our work.\footnote{\dd{Recall that SASRec uses BCE as a loss function - we do not test pairwise and listwise loss functions, because, as mentioned in Section~\ref{ssec:neg_sampling}, they are expensive to apply on GPUs, and (e.g.) LambdaRank~\cite{burgesRanknetLambdarankLambdamart2010} does not improve SASRec~\cite{PetrovRSS22, petrov2023rss}.}} We also use BERT4Rec as a state-of-the-art baseline.
\dd{For all models} we set the sequence length to 200.

Additionally, we use two simple baselines: a non-personalised popularity model, which always recommends the most popular items; and the classic Matrix Factorisation model with BPR~\cite{rendleBPRBayesianPersonalized2009} loss. 
Our implementation of SASRec (and \ourmodel{}) are based on the original 
code\footnote{\href{https://github.com/kang205/SASRec/}{https://github.com/kang205/SASRec/}}, 
whereas our implementation of BERT4Rec is based on the more efficient implementation\footnote{\href{https://github.com/asash/bert4rec\_repro}{https://github.com/asash/bert4rec\_repro}} from \scrc{our} recent reproducibility paper~\cite{Bert4RecRepro}. To ensure that the models are fully trained, we use an early stopping mechanism to stop training if NDCG@10 measured on the validation dataset has not improved for 200 epochs. 
\subsection{Results}
\subsubsection{RQ1. How does negative sampling affect BERT4Rec's performance gains over SASRec}\label{sssec:results:bert4rec_performance}

\looseness -1 To answer our first research question, we train both BERT4Rec and SASRec on the Steam and MovieLens-1M datasets using the sampling strategies, \srsd{which were originally used in these models}:
(i) one negative per positive and BCE loss (as in SASRec) and 
(ii) all negatives per positive and Softmax loss (as in BERT4Rec).\footnote{\srsd{In this RQ, our goal is to better understand BERT4Rec's gains over SASRec, so we only experiment with their original loss functions and sampling strategies; We apply other loss functions, such as Sampled Softmax, and more negative samples in Section~\ref{ssec:gbce_effect}}.} We use the original training objectives for both architectures: item masking in BERT4Rec and sequence shifting in SASRec; \stdtd{we also \sdd{retain} the architecture differences in the models (i.e.\ we keep uni-directional attention in SASRec and bi-directional attention from BERT4rec)}.
The results of our comparison are summarised in Table~\ref{table:sasvsbert}. The \sdd{magnitude of} SASRec and BERT4Rec results are aligned with those reported in~\cite{Bert4RecRepro}.
As can be seen from the table, in all four cases, changing \craig{of the} sampling strategy from the one used by SASRec to the one used in BERT4Rec significantly improves effectiveness. 
For example, SASRec's NDCG@10 on MovieLens-1M \craig{is} improved from 0.131 to 0.169 (+29.0\%) \stdtd{by removing negative sampling and applying Softmax loss.  BERT4Rec achieves a larger improvement of NDCG@10 on Steam (0.0513 $\rightarrow$ 0.0746: +45.4\%) when changing the sampling strategy from 1 negative to all negatives.}
In contrast, the effect of changing the architecture is moderate (e.g. statistically \craig{indistinguishable} in 2 out of 4 cases), and frequently negative (3 cases out of four, 1 significant). 

\begin{table}
\caption{Effects of model architecture and negative sampling on NDCG@10, for the MovieLens-1M (ML-1M) and Steam datasets. \stdtd{* denotes a significant change ($pvalue < 0.05$) in NDCG@10 caused by negative sampling (comparing horizontally) or model architecture (comparing vertically).}
}\label{table:sasvsbert}

\resizebox{\linewidth}{!}{

\begin{tabular}[b]{lllll}
    \toprule
                                            Dataset& \begin{tabular}[l]{@{}l@{}}Negative sampling \\
                                            and loss function$\rightarrow$\\ Architecture$\downarrow$\end{tabular} & \multicolumn{1}{l}{\begin{tabular}[l]{@{}l@{}}1 negative\\ per positive;\\ BCE Loss \\ (as SASRec)\end{tabular}} & \multicolumn{1}{c}{\begin{tabular}[l]{@{}l@{}}No negative \\sampling; \\ Softmax Loss \\ (as  BERT4Rec)\end{tabular}} & \multicolumn{1}{l}{\textbf{\begin{tabular}[l]{@{}l@{}}Negative \\ sampling \\ and loss \\effect\end{tabular}}} \\ \hline
                                            & SASRec                                                                   & 0.131                                                                     & 0.169                                                                        & {\color[HTML]{548235} \textbf{+29.0\%*}}                                                        \\
                                            & BERT4Rec                                                                 & 0.123                                                                     & 0.161                                                                        & {\color[HTML]{548235} \textbf{+30.8\%*}}                                                        \\
\multirow{-3}{*}{ML-1M}                     & \textbf{Architecture effect}                                                    & {\color[HTML]{FF0000} \textbf{-6.1\%}}                                    & {\color[HTML]{FF0000} \textbf{-4.7\%}}                                       & {\color[HTML]{548235} \textbf{}}                                                                \\ \hline
\multicolumn{1}{l}{}                        & SASRec                                                                   & 0.0581                                                                    & 0.0721                                                                       & {\color[HTML]{548235} \textbf{+24.1\%*}}                                                        \\
\multicolumn{1}{l}{}                        & BERT4Rec                                                                 & 0.0513                                                                    & 0.0746                                                                       & {\color[HTML]{548235} \textbf{+45.4\%*}}                                                        \\
\multicolumn{1}{l}{\multirow{-3}{*}{Steam}} & \textbf{Architecture effect}                                                    & {\color[HTML]{FF0000} \textbf{-11.7\%*}}                                  & {\color[HTML]{548235} \textbf{+3.4\%*}}                                      &                                                                                                 \\ \bottomrule
\end{tabular}
}
\end{table}

\looseness -1 In the answer to RQ1, we \craig{conclude} that the absence of negative sampling plays the key role in BERT4Rec's success over SASRec, whereas \craig{any gain by applying BERT4Rec's bi-directional attention architecture is only moderate and frequently negative}. Therefore, the performance gains of BERT4Rec over SASRec \dd{can} be attributed to the absence of negative sampling and Softmax loss and not to its architecture and training objective. \dd{This is contrary to the explanations of the original BERT4Rec authors in~\cite{BERT4Rec}, who attributed its superiority to its bi-directional attention mechanism (on the same datasets)}. \craig{We now analyse how gBCE changes the distribution of predicted probabilities.}

\subsubsection{RQ2. Effect of \ourloss{}~on predicted interaction probabilities.}\label{sssec:results:effect_of_ourloss}
\looseness -1 To analyse the effects of \ourloss{} on predicted probabilities, we train three models: a regular SASRec model and two configurations of ~\ourmodel{}: a first with 64 negatives and $t=0.5$ and a second with 256 negatives and $t=1.0$. Our goal is to compare prior probabilities $P(i)$ with probabilities predicted by the model $\hat{p}_i$. As we discuss in Section~\ref{ssec:background:sequential_recsys}, $P(i)$ is unknown, so direct measurement of such relation is hard. Hence, as a substitute for $P(i)$, we use the popular mean Precision@K metric, which according to Cormack et al.~\cite{cormackEstimatingPrecisionRandom1999} can be seen as a measurement of the conditional probability of an item being relevant, given its rank is less than K. We compare this metric with the average predicted probability of items \stdtd{retrieved at rank} less than K. We perform this comparison for cutoffs K in the range [1..100]. 
\srsf{Figure~\ref{fig:pred_actual} displays the comparison results for MovieLens-1M and Steam datasets, illustrating the expected theoretical relationship between Precision@K and Predicted Probability@K based on Theorem~\ref{theorem:bias}.}
\begin{figure*}[tb]
    \centering
    \subfloat[MovieLens-1M]{
        \includegraphics[height=3.5cm]{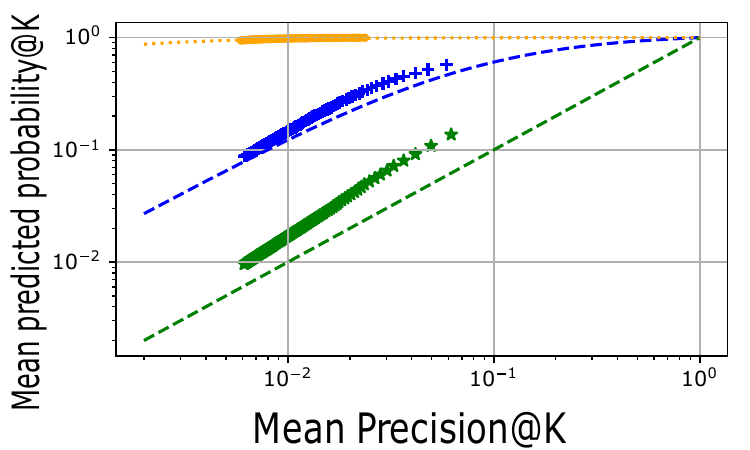}\label{subfig:pred_actual:ml1m}}
    \subfloat[Steam]{
        \includegraphics[height=3.5cm]{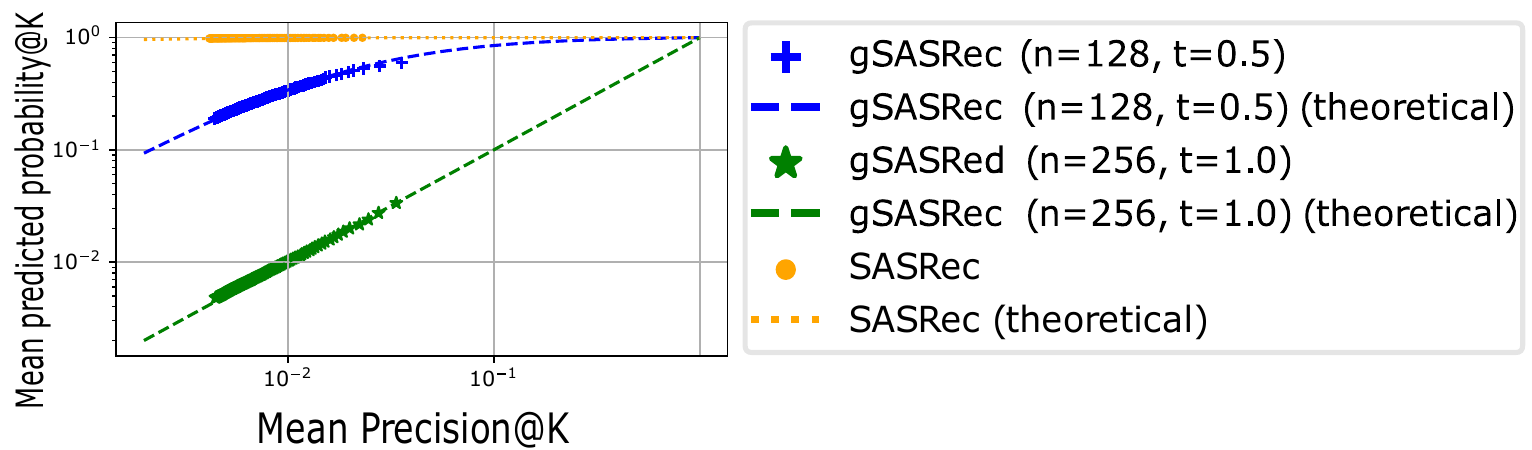}\label{subfig:pred_actual:steam}}
    \caption{Relation between Mean Precision@K metric and Mean predicted probability@K for cutoffs K in $[1..100]$ range.
    The figure also includes theoretical prediction for the relation according to Theorem~\ref{theorem:bias}.}\label{fig:pred_actual}
\end{figure*}

\looseness -1 \srs{Figure~\ref{subfig:pred_actual:steam} shows that the theoretical prediction from Theorem~\ref{theorem:bias} closely matches the observed relationship between Precision@K and Predicted Probability@K in the Steam dataset. In the MovieLens-1M dataset (Figure~\ref{subfig:pred_actual:ml1m}), a slight discrepancy appears between the theoretical prediction and observed relationship, likely because the smaller number of users in the dataset doesn't meet the requirement of Theorem~\ref{theorem:bias} for an adequate amount of training samples.}

\looseness -1 Despite these small discrepancies, the relation follows the trends expected from our theoretical analysis. In particular, Figure~\ref{fig:pred_actual} shows that as expected from Corollary~\ref{corollary:sasrec}, SASRec is indeed prone to overconfidence and on average predicts probabilities very close to 1 for all \stdtd{ranks less than 100}. In contrast, \sd{the probabilities predicted by \ourmodel{} are considerably less than 1}. For example, for MovieLens-1M, \ourmodel{}~trained with 128 negatives and $t=0.5$, on average predicts probability 0.57 at K=1, while the version with 256 negatives and $t=1.0$ predicts probability 0.13 at the same cutoff. \stdtd{Together, this analysis shows that \ourmodel{} trained with \ourloss{} successfully mitigates the overconfidence problem of SASRec.} %
Furthermore, from the figure we also see that when parameter $t$ is set to 1, the mean predicted probability is well-calibrated with mean precision at all \stdtd{rank cutoffs (particularly on the Steam dataset)}. This is well-aligned with Corollary~\ref{col:generalised_unbiased}, which states that when parameter $\beta$ in \ourloss{} is set equal to the sampling rate (i.e.\ setting parameter $t=1$) results in learning in fully calibrated probabilities. 
Overall in answer to RQ2, we conclude that \ourloss{}~successfully mitigates the overconfidence problem, and in a manner that is well-aligned with our theoretical analysis. \craig{We next turn to the impact of \ourloss{} on effectiveness.}

\subsubsection{RQ3. Effect of negative sampling rate and parameter $t$ on the performance of \ourmodel{}}

\begin{figure}
    \includegraphics[width=\linewidth]{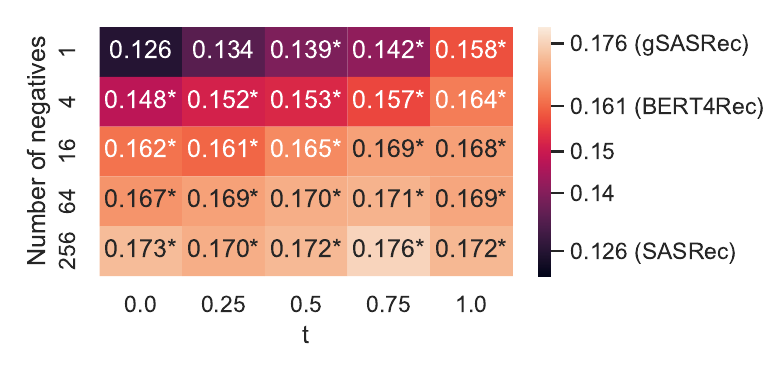}
    \caption{\ourmodel{}: Effect of varying number of negatives and \sfdtd{calibration} parameter $t$ on NDCG@10, MovieLens-1M. 
    * denotes a significant improvement over SASRec ($pvalue < 0.05$, Bonferroni multiple test correction).
    }\label{fig:gridsearch}
\end{figure}
\looseness -1 In comparison to SASRec, \ourmodel{} has two additional hyperparameters: the number of negative samples and the parameter $t$, which adjusts probability calibration. To explore the impact of these parameters on performance, we conduct a grid search: selecting the negative sample count from [1, 4, 16, 64, 256] and the calibration parameter $t$ from [0, 0.25, 0.5, 0.75, 1.0].

\looseness -1 Figure~\ref{fig:gridsearch}  \dd{portrays} our grid search on MovieLens-1M (on others datasets we observed \sdd{a} similar pattern and omit their figures for brevity). \craig{From the figure, we observe that}, as expected from the theoretical analysis, both $t$ and the number of negatives have a positive effect on model \sfdtd{effectiveness}. For example, when the number of negatives is set to 1, varying $t$ from 0 to 1 allows increasing NDCG@10 from 0.126 to 0.158 \stdtd{(+25\%, \sm{significant}, compared to SASRec, which is also \ourmodel{} with 1 negative and calibration $t=0$)}. Interestingly, the result of \ourmodel{}~with 1 negative and $t=1$ is similar to what BERT4Rec achieves with all negatives (\craig{0.158 vs.\ 0.161:} -1.86\%, not significant). We also observe that when the number of negatives is higher, setting \sdd{a} high value of $t$ is less important. For example, when the model is trained with 256 negatives (7.49\% sampling rate), the model achieves high effectiveness with all values of $t$. This is also not surprising -- by design, more negative samples and higher values of $t$ should have a similar effect \sdd{in \ourloss{}}. 
During our experiments, we also observed that setting parameter $t$ very close to 1 also increased the training time of the model. Keeping this in mind, in practical applications we recommend setting $t$ between 0.75 and 0.9, and the number of negatives between 128~and~256 -- this combination works well on all datasets, converging to results that are close to \craig{the best observed} without increasing training time. \craig{This answers RQ3.}

\begin{figure*}[tb]
    \centering
    \subfloat[Recall@1]{
        \includegraphics[width=0.5\textwidth]{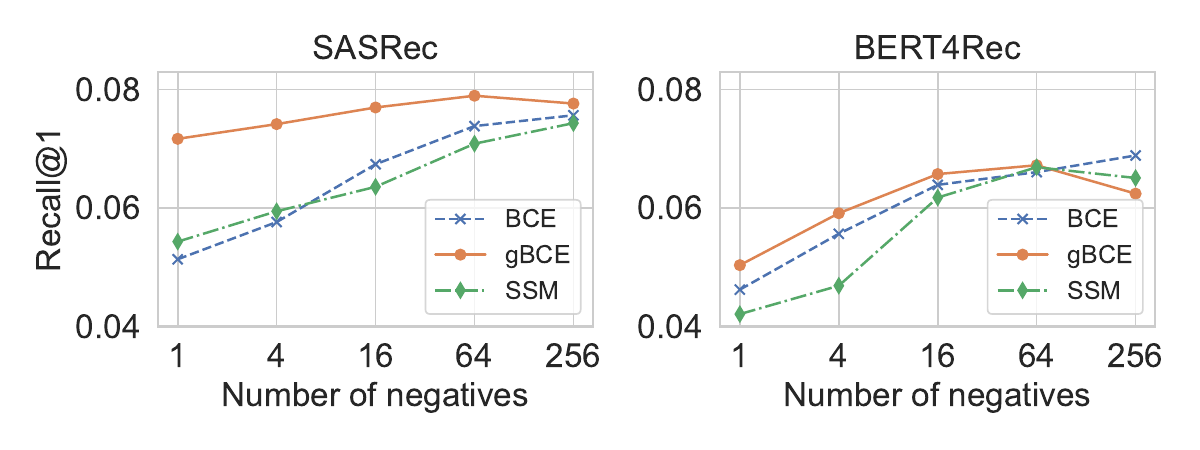}
    }
    \subfloat[NDCG@10]{
        \includegraphics[width=0.5\textwidth]{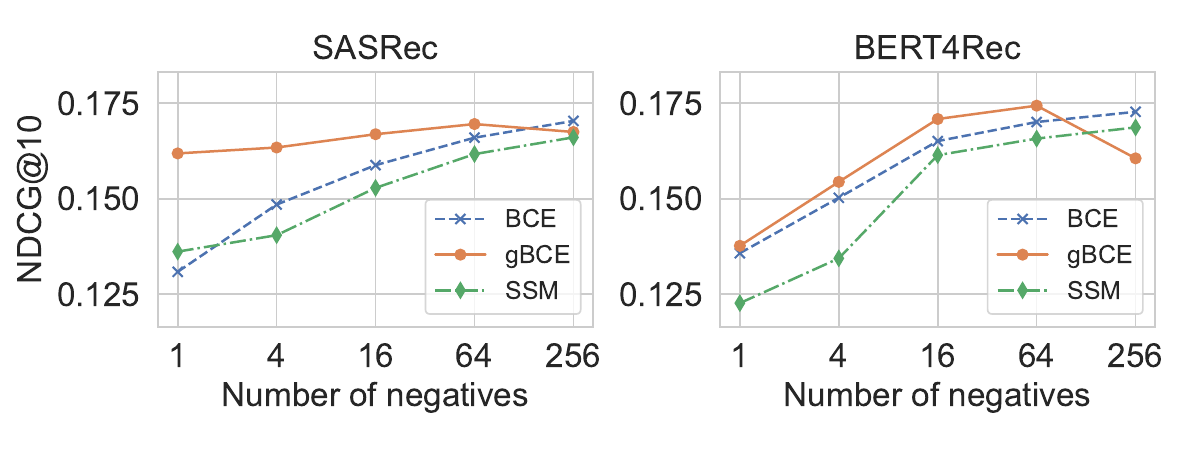}
    }
    \caption{Performance of SASRec and BERT4Rec architectures, trained on the MovieLens-1M dataset with a variable number of negatives and various loss functions. BCE is a classic Binary Cross Entropy Loss, gBCE - Generalised Binary Cross-entropy (t=1.0), SSM - Sampled Softmax Loss with uniform sampling. 
    }
    \label{fig:decoupled_loss}
\end{figure*}

\subsubsection{\srs{RQ4. Effect of gBCE loss on negatively sampled SASRec and BERT4Rec}}
\label{ssec:gbce_effect}
\srs{\srsf{To investigate the effect of gBCE on SASRec and BERT4Rec models with negative sampling
}, we train the models 
with the number of negative samples selected from [1, 4, 16, 64, 256] and the loss function selected from \srsf{[BCE, gBCE,  Sampled Softmax Loss (SSM)] on the MovieLens-1M dataset}. In this experiment, we use a fully-calibrated version of gBCE ($t=1.0$). 
Figure~\ref{fig:decoupled_loss} summarises the results of the experiment.}
\looseness -1 \srs{
     As we can see from the figure, gBCE performs better than both BCE and Sampled Softmax loss when the number of negatives is small. For example, for BERT4Rec trained with 4 negatives, gBCE has a higher Recall@1 (0.059) than both BCE (0.055; -5.8\% compared with gBCE) and Sampled Softmax (0.046, -20\%) and has the highest NDCG@10 of 0.154, while BCE has NDCG@10 of 0.150 (-2.6\% compared with gBCE) and Sampled Softmax has NDCG@10 of 0.134 (-12.9\%). 
    In the case of SASRec, the difference is even larger when the number of negatives is small (recall that SASRec trained with gBCE is also \ourmodel{}). For example, with 16 negatives, gBCE achieves Recall@1 0.0769, BCE achieves 0.0673 (-12.5\%), and Sampled Softmax achieves 0.0635 (-17.5\%).
   \srsf{We hypothesise that gBCE affects SASRec more than BERT4Rec due to their training objectives. SASRec predicts the next item in a sequence, while BERT4Rec predicts randomly masked items. Consequently, altering SASRec's loss function directly impacts its performance in next-item prediction. In contrast, changing BERT4Rec's loss function only affect the masking task which is less directly related to the next item prediction task~\cite{PetrovRSS22, petrov2023rss}.}}
\srs{
    \srsf{On the other hand, when more negatives are sampled, gBCE becomes less beneficial.} For example, with 256 negatives, all three loss functions achieve similar NDCG@10 (0.1674 gBCE; 0.1703 (+1.7\%)  BCE and 0.1660 (-0.08\%) Sampled Softmax). This is an expected result because 256 negatives represent a significant proportion of all negatives \srsf{(7.5\%)}, and overconfidence becomes less of an issue for BCE and Sampled Softmax.  
}

\rsm{In conclusion, for RQ4, gBCE outperforms BCE and Sampled Softmax in SASRec and BERT4Rec with few negatives; improvement is larger in SASRec. However, with many negatives, traditional loss functions like BCE and Sampled Softmax work well unaltered, but high sampling rates are impractical due to memory and computational constraints.}
 
\subsubsection{RQ5. \ourmodel{} performance in comparison to state-of the-art \sd{sequential} recommendation models.}
\label{ssec:sota_performance}
\looseness -1 To answer our last research question, we compare~\ourmodel{} with the baseline models. \srsd{We also add to comparison a version of SASRec trained with full Softmax  (without sampling) because, as we discuss in RQ1, it exhibits SOTA performance; however, we omit  non-standard versions of BERT4Rec and SASRec trained with BCE or Sampled Softmax, because as we report in RQ4 they are not beneficial compared to gBCE. We also exclude BERT4Rec with gBCE from our analysis because, as per RQ4, gSASRec achieves superior results when measured by Recall@1 and similar results when evaluated by NDCG@10.}
After tuning hyperparameters on the validation set, we report the results of gSASRec with 
128 negatives and $t=0.9$ for Steam and Gowalla, and with 256 negatives and $t=0.75$ for MovieLens-1M.  Table~\ref{tb:results} summarises the results of our evaluation. 
The table shows that \ourmodel{} achieved the best or the second best result \stdtd{on all datasets according to all metrics}.
Indeed, on the smaller datasets (MovieLens and \srs{Steam}), where we were able to train BERT4Rec and \scrc{SASRec} without sampling, \ourmodel{} performs similarly to the best unsampled model (e.g.\  +4.1\% NDCG@10 on MovieLens-1M compared to SASRec-Softmax (not significant) or -1.74\% compared to BERT4Rec on \scrc{Steam}, significant). Interestingly, on MovieLens-1M, both SASRec-Softmax (our version of SASRec trained without negative sampling) and \ourmodel{} \sm{significantly} improve Recall@1, suggesting that at least in some cases SASRec's unidirectional architecture may be beneficial. \stdtd{This also echoes our observations while analysing RQ1.}

\begin{table*}[t]
\caption{Evaluation results.  Bold denotes the best model on the dataset for that metric and underlined is the second-best model. *  denotes a significant difference with the best-performing model ($pvalue < 0.05$). \srs{SASRec-Softmax is a SASRec-based model trained without negative sampling and Softmax loss (as BERT4Rec)}.}\label{tb:results}
\resizebox{1.0\textwidth}{!}{
\begin{tabular}{|llll|llllllllllll|}
\hline
\multicolumn{4}{|c|}{\multirow{2}{*}{Model}}                                                                              & \multicolumn{12}{c|}{Datasets}                                                                                                                                                                                                                                                                                                                                                                                                                                                                                                                                                                                                                                                                                     \\ \cline{5-16}
\multicolumn{4}{|c|}{}                                                                                                    & \multicolumn{4}{c|}{MovieLens-1M}                                                                                                                                                                                                           & \multicolumn{4}{c|}{Steam}                                                                                                                                                                                                                  & \multicolumn{4}{c|}{Gowalla}                                                                                                                                                                                           \\ \hline
Category                   & Model          & \begin{tabular}[c]{@{}l@{}}Negative \\ Sampling\end{tabular} & Loss         & \begin{tabular}[c]{@{}l@{}}Recall\\ @1\end{tabular} & \begin{tabular}[c]{@{}l@{}}Recall\\ @10\end{tabular} & \begin{tabular}[c]{@{}l@{}}NDCG\\ @10\end{tabular} & \multicolumn{1}{l|}{\begin{tabular}[c]{@{}l@{}}Time\\ (min)\end{tabular}} & \begin{tabular}[c]{@{}l@{}}Recall\\ @1\end{tabular} & \begin{tabular}[c]{@{}l@{}}Recall\\ @10\end{tabular} & \begin{tabular}[c]{@{}l@{}}NDCG\\ @10\end{tabular} & \multicolumn{1}{l|}{\begin{tabular}[c]{@{}l@{}}Time\\ (min)\end{tabular}} & \begin{tabular}[c]{@{}l@{}}Recall\\ @1\end{tabular} & \begin{tabular}[c]{@{}l@{}}Recall\\ @10\end{tabular} & \begin{tabular}[c]{@{}l@{}}NDCG\\ @10\end{tabular} & \begin{tabular}[c]{@{}l@{}}Time\\ (min)\end{tabular} \\ \hline
\multirow{2}{*}{Baselines} & Popularity     & N/A                                                          & N/A          & 0.005*                                              & 0.036*                                               & 0.017*                                             & \multicolumn{1}{l|}{0.0}                                                  & 0.0077*                                             & 0.0529*                                              & 0.0268*                                            & \multicolumn{1}{l|}{0.0}                                                  & 0.0011*                                             & 0.0081*                                              & 0.0041*                                            & 0.1                                                  \\
                           & MF-BPR         & Yes                                                          & BPR          & 0.010*                                              & 0.075*                                               & 0.037*                                             & \multicolumn{1}{l|}{0.1}                                                  & 0.0071*                                             & 0.0393*                                              & 0.0206*                                            & \multicolumn{1}{l|}{0.4}                                                  & 0.0083*                                             & 0.0282*                                              & 0.0170*                                            & 2.1                                                  \\ \hline
\multirow{2}{*}{Unsampled} & BERT4Rec       & No                                                           & Softmax & 0.058*                                              & \underline{0.294}                                          & 0.161*                                             & \multicolumn{1}{l|}{86}                                                   & \underline{0.0281}                                        & \textbf{0.1379}                                      & \textbf{0.0746}                                    & \multicolumn{1}{l|}{642}                                                  & \multicolumn{4}{l|}{\multirow{2}{*}{\begin{tabular}[c]{@{}l@{}}Infeasible (requires $>$100GB of\\ GPU memory, see Section~\ref{sec:intro})\end{tabular}}}                                                                                           \\      
                           & SASRec-Softmax & No                                                           & Softmax & \underline{0.073}                                         & 0.293                                                & \underline{0.169}                                        & \multicolumn{1}{l|}{9}                                                    & 0.0280                                              & 0.1323*                                              & 0.0721*                                            & \multicolumn{1}{l|}{80}                                                   & \multicolumn{4}{l|}{}                                                                                                                                                                                                  \\ \hline
\multirow{2}{*}{Sampled}   & SASRec         & Yes                                                          & BCE          & 0.046*                                              & 0.247*                                               & 0.131*                                             & \multicolumn{1}{l|}{13}                                                   & 0.0193*                                             & 0.1121*                                              & 0.0581*                                            & \multicolumn{1}{l|}{32}                                                   & \underline{0.0505*}                                       & \underline{0.1831*}                                        & \underline{0.1097*}                                      & 145                                                  \\
                           & gSASRec        & Yes                                                          & gBCE         & \textbf{0.082}                                      & \textbf{0.300}                                       & \textbf{0.176}                                     & \multicolumn{1}{l|}{23}                                                   & \textbf{0.0283}                                     & \underline{0.1355}                                         & \underline{0.0735}                                       & \multicolumn{1}{l|}{58}                                                   & \textbf{0.0782}                                     & \textbf{0.2590}                                      & \textbf{0.1616}                                    & 191                                                  \\ \hline
\end{tabular}
}
\end{table*}
\looseness -1 Crucially, \ourmodel{}~always \sm{significantly} outperforms the regular SASRec model (+34\% NDCG@10 on MovieLens-1M, +26\% on Steam, +47\% on Gowalla). The result on Gowalla is particularly important, as it demonstrates that \ourmodel{} is suitable for datasets with more than 1 million items, and it improves SASRec's results by a large margin on \sdd{this large dataset}. 
\sm{From Table~\ref{tb:results} we also see that all versions of SASRec (including \ourmodel{}) require less training time than BERT4Rec. For example, on \sdd{MovieLens,} \ourmodel{} is 73.2\% faster to train compared to \scrc{BERT4Rec} (23 minutes vs.\ 85 minutes) and, on Steam, \ourmodel{} is 90.9\% faster (58 minutes vs.\  642 minutes). \srsf{However, we also see that gSASRec requires more training time than SASRec (e.g.\ 58 vs.\ 32 minutes on Steam); we explain that by the fact that more accurate probabilities estimation with gBCE requires more training epochs to converge (238 epochs vs.\ 170 epochs in our experiment).}}

\balance
\looseness -1 Finally,  for MovieLens-1M, we compare the results achieved by \ourmodel{} with \dd{those} of the most recently proposed models in the literature, which report the best results, namely ALBERT4Rec~\cite{Bert4RecRepro} (an effective model similar to BERT4Rec, but based on ALBERT~\cite{ALBERT}), and two contrastive models:  \stdtd{DuoRec~\cite{DuoRec}, and CBiT~\cite{CBiT}}. All papers from our selection \sdd{the} use the same data-splitting strategy and unsampled metrics, so the results are comparable. Table~\ref{tb:reported_sota} summarises this comparison. As we can see from the table, all these publications report Recall@10 close to 0.3, which is similar to what we obtain with \ourmodel{}. \scrc{However, only \ourmodel{} achieves \scrci{an} NDCG@10 above 0.17. Furthermore, as observed from Figure~\ref{fig:gridsearch}, this result is not a one-off occurrence but a consistent outcome when the model is trained with 256 negatives, making it unlikely to be a statistical fluctuation.} This is likely due to \dd{its} better focus on highly-ranked items, as~\ourloss{} is specifically designed for improving overconfidence in highly-scored items. 

\looseness -1 Overall, our experiments show  that \ourmodel{}~performs on par with SOTA models, retaining the negative sampling required for use on big catalogues \sm{and converging faster than BERT4Rec}.
\balance
\section{Conclusions}\label{sec:conclusion}
\begin{table}
\caption{Comparison of gSASRec with recent reported results on MovieLens-1M. Bold indicates the best value.}\label{tb:reported_sota}
\centering
\begin{tabular}{lll}
\toprule
                                                      & Recall@10      & NDCG@10        \\ \midrule
DuoRec~\cite{DuoRec}            & 0.294          & 0.168          \\
ALBERT4Rec~\cite{Bert4RecRepro} & 0.300          & 0.165          \\
CBiT~\cite{CBiT}               & \textbf{0.301} & 0.169          \\
gSASRec                                               & 0.300          & \textbf{0.176} \\\bottomrule
\end{tabular}
\end{table}
\looseness -1 In this paper, we \scrci{studied} the \scrci{impact} of negative sampling on sequential recommender systems. We showed (theoretically and empirically) that negative sampling coupled with Binary Cross-Entropy loss (a popular combination used by many sequential models) leads to \craig{a shifted score distribution, called overconfidence}.
We showed that overconfidence is the only reason why SASRec underperforms compared to \scrci{the state-of-the-art} BERT4Rec. \scrci{Indeed,} when we control for negative sampling, the two models perform similarly. We proposed \sdd{a solution} to the \scrci{overconfidence problem} in \scrci{the} form of \ourloss{}~and theoretically \sdd{proved} that it can mitigate overconfidence. We further proposed \ourmodel{}, which uses~\sdd{\ourloss{},} and experimentally showed that \sdd{\scrci{it can significantly} outperform} the best unsampled models \sdd{(e.g.\ +9.47\% NDCG@10 on \scrci{the} MovieLens-1M \scrci{dataset} compared to BERT4Rec) requiring less training time (e.g.\ -90.9\% on \scrci{the} Steam \scrci{dataset} compared to BERT4Rec)}, while \dd{also being} suitable for large-scale datasets. \srsf{We also showed that \ourloss{} may be beneficial for BERT4Rec if it is trained with negative sampling (e.g. +7.2\% compared to BCE when trained with 4 negatives).}
\dd{The theory and methods presented in this paper \scrci{could} also be \scrci{applied} to not just sequential recommendation models but also \scrci{to} other types of recommendation \scrci{as well as} to NLP or search tasks -- we leave these directions to future work.}
\FloatBarrier
\bibliographystyle{ACM-Reference-Format}
\balance
\bibliography{references}


\begin{thebibliography}{46}

%%% the Web address.

\ifx \showCODEN    \undefined \def \showCODEN     #1{\unskip}     \fi
\ifx \showDOI      \undefined \def \showDOI       #1{#1}\fi
\ifx \showISBNx    \undefined \def \showISBNx     #1{\unskip}     \fi
\ifx \showISBNxiii \undefined \def \showISBNxiii  #1{\unskip}     \fi
\ifx \showISSN     \undefined \def \showISSN      #1{\unskip}     \fi
\ifx \showLCCN     \undefined \def \showLCCN      #1{\unskip}     \fi
\ifx \shownote     \undefined \def \shownote      #1{#1}          \fi
\ifx \showarticletitle \undefined \def \showarticletitle #1{#1}   \fi
\ifx \showURL      \undefined \def \showURL       {\relax}        \fi
% invisible to TeX
\providecommand\bibfield[2]{#2}
\providecommand\bibinfo[2]{#2}
\providecommand\natexlab[1]{#1}
\providecommand\showeprint[2][]{arXiv:#2}

\bibitem[Acquavia et~al\mbox{.}(2023)]%
        {acquavia2023static}
\bibfield{author}{\bibinfo{person}{Antonio Acquavia}, \bibinfo{person}{Nicola
  Tonellotto}, {and} \bibinfo{person}{Craig Macdonald}.}
  \bibinfo{year}{2023}\natexlab{}.
\newblock \showarticletitle{Static Pruning for Multi-Representation Dense
  Retrieval}. In \bibinfo{booktitle}{\emph{Proc. {{DocEng}}}}.
\newblock


\bibitem[Burges(2010)]%
        {burgesRanknetLambdarankLambdamart2010}
\bibfield{author}{\bibinfo{person}{Christopher Burges}.}
  \bibinfo{year}{2010}\natexlab{}.
\newblock \showarticletitle{From Rank{N}et to Lambda{R}ank to Lambda{MART}:
  {{An}} Overview}.
\newblock \bibinfo{journal}{\emph{Learning}}  \bibinfo{volume}{11}
  (\bibinfo{date}{Jan.} \bibinfo{year}{2010}).
\newblock


\bibitem[Ca{\~n}amares and Castells(2020)]%
        {canamaresTargetItemSampling2020}
\bibfield{author}{\bibinfo{person}{Roc{\'i}o Ca{\~n}amares} {and}
  \bibinfo{person}{Pablo Castells}.} \bibinfo{year}{2020}\natexlab{}.
\newblock \showarticletitle{On {{Target Item Sampling}} in {{Offline
  Recommender System Evaluation}}}. In \bibinfo{booktitle}{\emph{Proc.
  {{RecSys}}}}. \bibinfo{pages}{259--268}.
\newblock


\bibitem[Chen et~al\mbox{.}({[n.\,d.]})]%
        {chenGeneratingNegativeSamples}
\bibfield{author}{\bibinfo{person}{Yongjun Chen}, \bibinfo{person}{Jia Li},
  \bibinfo{person}{Zhiwei Liu}, \bibinfo{person}{Nitish~Shirish Keskar},
  \bibinfo{person}{Huan Wang}, \bibinfo{person}{Julian McAuley}, {and}
  \bibinfo{person}{Caiming Xiong}.} \bibinfo{year}{[n.\,d.]}\natexlab{}.
\newblock \bibinfo{title}{Generating {{Negative Samples}} for {{Sequential
  Recommendation}}}.
\newblock
\newblock
\showeprint[arxiv]{2208.03645}~[cs]


\bibitem[Cho et~al\mbox{.}(2011)]%
        {Gowalla}
\bibfield{author}{\bibinfo{person}{Eunjoon Cho}, \bibinfo{person}{Seth~A.
  Myers}, {and} \bibinfo{person}{Jure Leskovec}.}
  \bibinfo{year}{2011}\natexlab{}.
\newblock \showarticletitle{Friendship and Mobility: User Movement in
  Location-Based Social Networks}. In \bibinfo{booktitle}{\emph{Proc.
  {{KDD}}}}. \bibinfo{pages}{1082}.
\newblock


\bibitem[Cormack et~al\mbox{.}(1999)]%
        {cormackEstimatingPrecisionRandom1999}
\bibfield{author}{\bibinfo{person}{Gordon~V. Cormack}, \bibinfo{person}{Ondrej
  Lhotak}, {and} \bibinfo{person}{Christopher~R. Palmer}.}
  \bibinfo{year}{1999}\natexlab{}.
\newblock \showarticletitle{Estimating Precision by Random Sampling}. In
  \bibinfo{booktitle}{\emph{Proc. {{SIGIR}}}}. \bibinfo{pages}{273--274}.
\newblock


\bibitem[Dallmann et~al\mbox{.}(2021)]%
        {dallmannCaseStudySampling2021}
\bibfield{author}{\bibinfo{person}{Alexander Dallmann}, \bibinfo{person}{Daniel
  Zoller}, {and} \bibinfo{person}{Andreas Hotho}.}
  \bibinfo{year}{2021}\natexlab{}.
\newblock \showarticletitle{A {{Case Study}} on {{Sampling Strategies}} for
  {{Evaluating Neural Sequential Item Recommendation Models}}}. In
  \bibinfo{booktitle}{\emph{Proc. {{RecSys}}}}. \bibinfo{pages}{505--514}.
\newblock


\bibitem[Devlin et~al\mbox{.}(2019)]%
        {BERT}
\bibfield{author}{\bibinfo{person}{Jacob Devlin}, \bibinfo{person}{Ming-Wei
  Chang}, \bibinfo{person}{Kenton Lee}, {and} \bibinfo{person}{Kristina
  Toutanova}.} \bibinfo{year}{2019}\natexlab{}.
\newblock \showarticletitle{{{BERT}}: {{Pre-training}} of {{Deep Bidirectional
  Transformers}} for {{Language Understanding}}}. In
  \bibinfo{booktitle}{\emph{Proc. of {{NAACL-HLT}}}}.
  \bibinfo{pages}{4171--4186}.
\newblock


\bibitem[Dodge et~al\mbox{.}(2021)]%
        {dodgeDocumentingLargeWebtext2021a}
\bibfield{author}{\bibinfo{person}{Jesse Dodge}, \bibinfo{person}{Maarten Sap},
  \bibinfo{person}{Ana Marasovi{\'c}}, \bibinfo{person}{William Agnew},
  \bibinfo{person}{Gabriel Ilharco}, \bibinfo{person}{Dirk Groeneveld},
  \bibinfo{person}{Margaret Mitchell}, {and} \bibinfo{person}{Matt Gardner}.}
  \bibinfo{year}{2021}\natexlab{}.
\newblock \showarticletitle{Documenting {{Large Webtext Corpora}}: {{A Case
  Study}} on the {{Colossal Clean Crawled Corpus}}}. In
  \bibinfo{booktitle}{\emph{Proc. {{EMNLP}}}}. \bibinfo{pages}{1286--1305}.
\newblock


\bibitem[Du et~al\mbox{.}(2022)]%
        {CBiT}
\bibfield{author}{\bibinfo{person}{Hanwen Du}, \bibinfo{person}{Hui Shi},
  \bibinfo{person}{Pengpeng Zhao}, \bibinfo{person}{Deqing Wang},
  \bibinfo{person}{Victor~S. Sheng}, \bibinfo{person}{Yanchi Liu},
  \bibinfo{person}{Guanfeng Liu}, {and} \bibinfo{person}{Lei Zhao}.}
  \bibinfo{year}{2022}\natexlab{}.
\newblock \showarticletitle{Contrastive {{Learning}} with {{Bidirectional
  Transformers}} for {{Sequential Recommendation}}}. In
  \bibinfo{booktitle}{\emph{Proc. {{CIKM}}}}. \bibinfo{pages}{396--405}.
\newblock


\bibitem[Goodfellow et~al\mbox{.}(2016)]%
        {goodfellowDeepLearning2016}
\bibfield{author}{\bibinfo{person}{Ian Goodfellow}, \bibinfo{person}{Yoshua
  Bengio}, {and} \bibinfo{person}{Aaron Courville}.}
  \bibinfo{year}{2016}\natexlab{}.
\newblock \bibinfo{booktitle}{\emph{Deep {{Learning}}}}.
\newblock \bibinfo{publisher}{{MIT Press}}.
\newblock


\bibitem[Harper and Konstan(2015)]%
        {harperMovieLensDatasetsHistory2015a}
\bibfield{author}{\bibinfo{person}{F.~Maxwell Harper} {and}
  \bibinfo{person}{Joseph~A. Konstan}.} \bibinfo{year}{2015}\natexlab{}.
\newblock \showarticletitle{The {{MovieLens Datasets}}: {{History}} and
  {{Context}}}.
\newblock \bibinfo{journal}{\emph{ACM Transactions on Interactive Intelligent
  Systems (TiiS)}} \bibinfo{volume}{5}, \bibinfo{number}{4}
  (\bibinfo{date}{Dec.} \bibinfo{year}{2015}), \bibinfo{pages}{19:1--19:19}.
\newblock
\showISSN{2160-6455}


\bibitem[Heaps(1978)]%
        {heapsInformationRetrievalComputational1978}
\bibfield{author}{\bibinfo{person}{H.~S. Heaps}.}
  \bibinfo{year}{1978}\natexlab{}.
\newblock \bibinfo{booktitle}{\emph{Information {{Retrieval}}:
  {{Computational}} and {{Theoretical Aspects}}}}.
\newblock \bibinfo{publisher}{{Academic Press, Inc.}},
  \bibinfo{address}{{USA}}.
\newblock
\showISBNx{978-0-12-335750-2}


\bibitem[Hidasi et~al\mbox{.}(2016)]%
        {GRU4Rec}
\bibfield{author}{\bibinfo{person}{Bal{\'a}zs Hidasi},
  \bibinfo{person}{Alexandros Karatzoglou}, \bibinfo{person}{Linas Baltrunas},
  {and} \bibinfo{person}{Domonkos Tikk}.} \bibinfo{year}{2016}\natexlab{}.
\newblock \showarticletitle{Session-Based {{Recommendations}} with {{Recurrent
  Neural Networks}}}. In \bibinfo{booktitle}{\emph{Proc. {{ICLR}}}}.
\newblock


\bibitem[Jean et~al\mbox{.}(2015)]%
        {jeanUsingVeryLarge2015}
\bibfield{author}{\bibinfo{person}{S{\'e}bastien Jean},
  \bibinfo{person}{Kyunghyun Cho}, \bibinfo{person}{Roland Memisevic}, {and}
  \bibinfo{person}{Yoshua Bengio}.} \bibinfo{year}{2015}\natexlab{}.
\newblock \showarticletitle{On {{Using Very Large Target Vocabulary}} for
  {{Neural Machine Translation}}}. In \bibinfo{booktitle}{\emph{Proc.
  {{ACL-IJCNLP}}}}. \bibinfo{pages}{1--10}.
\newblock


\bibitem[Ji et~al\mbox{.}(2020)]%
        {jiRevisitPopularityBaseline2020}
\bibfield{author}{\bibinfo{person}{Yitong Ji}, \bibinfo{person}{Aixin Sun},
  \bibinfo{person}{Jie Zhang}, {and} \bibinfo{person}{Chenliang Li}.}
  \bibinfo{year}{2020}\natexlab{}.
\newblock \showarticletitle{A {{Re-visit}} of the {{Popularity Baseline}} in
  {{Recommender Systems}}}. In \bibinfo{booktitle}{\emph{Proc. {{SIGIR}}}}.
  \bibinfo{pages}{1749--1752}.
\newblock


\bibitem[Kang and McAuley(2018)]%
        {SASRec}
\bibfield{author}{\bibinfo{person}{Wang-Cheng Kang} {and}
  \bibinfo{person}{Julian McAuley}.} \bibinfo{year}{2018}\natexlab{}.
\newblock \showarticletitle{Self-{{Attentive Sequential Recommendation}}}. In
  \bibinfo{booktitle}{\emph{Proc. {{ICDM}}}}. \bibinfo{pages}{197--206}.
\newblock
\showISSN{2374-8486}


\bibitem[Krichene and Rendle(2022)]%
        {kricheneSampledMetricsItem2022}
\bibfield{author}{\bibinfo{person}{Walid Krichene} {and}
  \bibinfo{person}{Steffen Rendle}.} \bibinfo{year}{2022}\natexlab{}.
\newblock \showarticletitle{On Sampled Metrics for Item Recommendation}.
\newblock \bibinfo{journal}{\emph{Commun. ACM}} \bibinfo{volume}{65},
  \bibinfo{number}{7} (\bibinfo{date}{June} \bibinfo{year}{2022}),
  \bibinfo{pages}{75--83}.
\newblock


\bibitem[Lan et~al\mbox{.}(2020)]%
        {ALBERT}
\bibfield{author}{\bibinfo{person}{Zhenzhong Lan}, \bibinfo{person}{Mingda
  Chen}, \bibinfo{person}{Sebastian Goodman}, \bibinfo{person}{Kevin Gimpel},
  \bibinfo{person}{Piyush Sharma}, {and} \bibinfo{person}{Radu Soricut}.}
  \bibinfo{year}{2020}\natexlab{}.
\newblock \showarticletitle{{{ALBERT}}: {{A Lite BERT}} for {{Self-supervised
  Learning}} of {{Language Representations}}}. In
  \bibinfo{booktitle}{\emph{Proc. {{ICLR}}}}.
\newblock


\bibitem[Lin et~al\mbox{.}(2022)]%
        {linPretrainedTransformersText2022}
\bibfield{author}{\bibinfo{person}{Jimmy Lin}, \bibinfo{person}{Rodrigo
  Nogueira}, {and} \bibinfo{person}{Andrew Yates}.}
  \bibinfo{year}{2022}\natexlab{}.
\newblock \bibinfo{booktitle}{\emph{Pretrained {{Transformers}} for {{Text
  Ranking}}: {{BERT}} and {{Beyond}}}}.
\newblock \bibinfo{publisher}{{Springer International Publishing}}.
\newblock


\bibitem[Murphy(2022)]%
        {murphyProbabilisticMachineLearning2022}
\bibfield{author}{\bibinfo{person}{Kevin~P. Murphy}.}
  \bibinfo{year}{2022}\natexlab{}.
\newblock \bibinfo{booktitle}{\emph{Probabilistic Machine Learning: An
  Introduction}}.
\newblock \bibinfo{publisher}{{The MIT Press}}, \bibinfo{address}{{Cambridge,
  Massachusetts}}.
\newblock
\showISBNx{978-0-262-04682-4}
\showLCCN{Q325.5 .M872 2022}


\bibitem[Pathak et~al\mbox{.}(2017)]%
        {pathakGeneratingPersonalizingBundle2017}
\bibfield{author}{\bibinfo{person}{Apurva Pathak}, \bibinfo{person}{Kshitiz
  Gupta}, {and} \bibinfo{person}{Julian McAuley}.}
  \bibinfo{year}{2017}\natexlab{}.
\newblock \showarticletitle{Generating and {{Personalizing Bundle
  Recommendations}} on {{{\emph{Steam}}}}}. In \bibinfo{booktitle}{\emph{Proc.
  {{SIGIR}}}}. \bibinfo{pages}{1073--1076}.
\newblock


\bibitem[Pellegrini et~al\mbox{.}(2022)]%
        {pellegriniDonRecommendObvious2022a}
\bibfield{author}{\bibinfo{person}{Roberto Pellegrini}, \bibinfo{person}{Wenjie
  Zhao}, {and} \bibinfo{person}{Iain Murray}.} \bibinfo{year}{2022}\natexlab{}.
\newblock \showarticletitle{Don't Recommend the Obvious: Estimate Probability
  Ratios}. In \bibinfo{booktitle}{\emph{Proc. {{RecSys}}}}.
  \bibinfo{pages}{188--197}.
\newblock


\bibitem[Petrov and Macdonald(2022a)]%
        {PetrovRSS22}
\bibfield{author}{\bibinfo{person}{Aleksandr Petrov} {and}
  \bibinfo{person}{Craig Macdonald}.} \bibinfo{year}{2022}\natexlab{a}.
\newblock \showarticletitle{Effective and {{Efficient Training}} for
  {{Sequential Recommendation}} Using {{Recency Sampling}}}. In
  \bibinfo{booktitle}{\emph{Proc. {{RecSys}}}}. \bibinfo{pages}{81--91}.
\newblock


\bibitem[Petrov and Macdonald(2022b)]%
        {Bert4RecRepro}
\bibfield{author}{\bibinfo{person}{Aleksandr Petrov} {and}
  \bibinfo{person}{Craig Macdonald}.} \bibinfo{year}{2022}\natexlab{b}.
\newblock \showarticletitle{A {{Systematic Review}} and {{Replicability Study}}
  of {{BERT4Rec}} for {{Sequential Recommendation}}}. In
  \bibinfo{booktitle}{\emph{Proc. {{RecSys}}}}. \bibinfo{pages}{436--447}.
\newblock


\bibitem[Petrov and Macdonald(2023a)]%
        {petrov2023rss}
\bibfield{author}{\bibinfo{person}{Aleksandr Petrov} {and}
  \bibinfo{person}{Craig Macdonald}.} \bibinfo{year}{2023}\natexlab{a}.
\newblock \showarticletitle{RSS: Effective and Efficient Training for
  Sequential Recommendation Using Recency Sampling}.
\newblock \bibinfo{journal}{\emph{ACM Transactions on Recommender Systems
  (TORS)}} (\bibinfo{year}{2023}).
\newblock


\bibitem[Petrov and Macdonald(2023b)]%
        {petrov2023generative}
\bibfield{author}{\bibinfo{person}{Aleksandr~V Petrov} {and}
  \bibinfo{person}{Craig Macdonald}.} \bibinfo{year}{2023}\natexlab{b}.
\newblock \showarticletitle{Generative Sequential Recommendation with GPTRec}.
  In \bibinfo{booktitle}{\emph{Proc. GenIR@SIGIR}}.
\newblock


\bibitem[Pham et~al\mbox{.}(2019)]%
        {phamCombinationAnalyticSignal2019}
\bibfield{author}{\bibinfo{person}{Luan~Thanh Pham}, \bibinfo{person}{Erdinc
  Oksum}, \bibinfo{person}{Thanh~Duc Do}, \bibinfo{person}{Minh {Le-Huy}},
  \bibinfo{person}{Minh~Duc Vu}, {and} \bibinfo{person}{Vinh~Duc Nguyen}.}
  \bibinfo{year}{2019}\natexlab{}.
\newblock \showarticletitle{{{LAS}}: {{A}} Combination of the Analytic Signal
  Amplitude and the Generalised Logistic Function as a Novel Edge Enhancement
  of Magnetic Data}.
\newblock \bibinfo{journal}{\emph{Contributions to Geophysics \& Geodesy}}
  \bibinfo{volume}{49}, \bibinfo{number}{4} (\bibinfo{year}{2019}),
  \bibinfo{pages}{425--440}.
\newblock


\bibitem[Qiu et~al\mbox{.}(2022)]%
        {DuoRec}
\bibfield{author}{\bibinfo{person}{Ruihong Qiu}, \bibinfo{person}{Zi Huang},
  \bibinfo{person}{Hongzhi Yin}, {and} \bibinfo{person}{Zijian Wang}.}
  \bibinfo{year}{2022}\natexlab{}.
\newblock \showarticletitle{Contrastive {{Learning}} for {{Representation
  Degeneration Problem}} in {{Sequential Recommendation}}}. In
  \bibinfo{booktitle}{\emph{Proc. {{WSDM}}}}. \bibinfo{pages}{813--823}.
\newblock


\bibitem[Rendle(2022)]%
        {rendleItemRecommendationImplicit2022}
\bibfield{author}{\bibinfo{person}{Steffen Rendle}.}
  \bibinfo{year}{2022}\natexlab{}.
\newblock \showarticletitle{Item {{Recommendation}} from {{Implicit
  Feedback}}}.
\newblock In \bibinfo{booktitle}{\emph{Recommender {{Systems Handbook}}}}.
  \bibinfo{pages}{143--171}.
\newblock


\bibitem[Rendle and Freudenthaler(2014)]%
        {rendleImprovingPairwiseLearning2014}
\bibfield{author}{\bibinfo{person}{Steffen Rendle} {and}
  \bibinfo{person}{Christoph Freudenthaler}.} \bibinfo{year}{2014}\natexlab{}.
\newblock \showarticletitle{Improving Pairwise Learning for Item Recommendation
  from Implicit Feedback}. In \bibinfo{booktitle}{\emph{Proc. {{WSDM}}}}.
\newblock


\bibitem[Rendle et~al\mbox{.}(2009)]%
        {rendleBPRBayesianPersonalized2009}
\bibfield{author}{\bibinfo{person}{Steffen Rendle}, \bibinfo{person}{Christoph
  Freudenthaler}, \bibinfo{person}{Zeno Gantner}, {and} \bibinfo{person}{Lars
  {Schmidt-Thieme}}.} \bibinfo{year}{2009}\natexlab{}.
\newblock \showarticletitle{{{BPR}}: {{Bayesian Personalized Ranking}} from
  {{Implicit Feedback}}}. In \bibinfo{booktitle}{\emph{Proc. {{UAI}}}}.
\newblock


\bibitem[Rendle et~al\mbox{.}(2010)]%
        {FPMC}
\bibfield{author}{\bibinfo{person}{Steffen Rendle}, \bibinfo{person}{Christoph
  Freudenthaler}, {and} \bibinfo{person}{Lars {Schmidt-Thieme}}.}
  \bibinfo{year}{2010}\natexlab{}.
\newblock \showarticletitle{Factorizing Personalized {{Markov}} Chains for
  Next-Basket Recommendation}. In \bibinfo{booktitle}{\emph{Proc. {{WWW}}}}.
  \bibinfo{pages}{811}.
\newblock


\bibitem[Richards(1959)]%
        {richardsFlexibleGrowthFunction1959}
\bibfield{author}{\bibinfo{person}{F.~J. Richards}.}
  \bibinfo{year}{1959}\natexlab{}.
\newblock \showarticletitle{A {{Flexible Growth Function}} for {{Empirical
  Use}}}.
\newblock \bibinfo{journal}{\emph{Journal of Experimental Botany}}
  \bibinfo{volume}{10}, \bibinfo{number}{2} (\bibinfo{year}{1959}),
  \bibinfo{pages}{290--301}.
\newblock
\showISSN{0022-0957, 1460-2431}


\bibitem[Stolcke(1998)]%
        {stolckeEntropybasedPruningBackoff1998}
\bibfield{author}{\bibinfo{person}{A. Stolcke}.}
  \bibinfo{year}{1998}\natexlab{}.
\newblock \showarticletitle{Entropy-Based {{Pruning}} of {{Backoff Language
  Models}}}. In \bibinfo{booktitle}{\emph{Proc. {{Broadcast News Tanscription}}
  and {{Understanding Workshop}}}}. \bibinfo{pages}{270--274}.
\newblock


\bibitem[Sun et~al\mbox{.}(2019)]%
        {BERT4Rec}
\bibfield{author}{\bibinfo{person}{Fei Sun}, \bibinfo{person}{Jun Liu},
  \bibinfo{person}{Jian Wu}, \bibinfo{person}{Changhua Pei},
  \bibinfo{person}{Xiao Lin}, \bibinfo{person}{Wenwu Ou}, {and}
  \bibinfo{person}{Peng Jiang}.} \bibinfo{year}{2019}\natexlab{}.
\newblock \showarticletitle{{{BERT4Rec}}: {{Sequential Recommendation}} with
  {{Bidirectional Encoder Representations}} from {{Transformer}}}. In
  \bibinfo{booktitle}{\emph{Proc. {{CIKM}}}}. \bibinfo{pages}{1441--1450}.
\newblock


\bibitem[Tang and Wang(2018)]%
        {Caser}
\bibfield{author}{\bibinfo{person}{Jiaxi Tang} {and} \bibinfo{person}{Ke
  Wang}.} \bibinfo{year}{2018}\natexlab{}.
\newblock \showarticletitle{Personalized {{Top-N Sequential Recommendation}}
  via {{Convolutional Sequence Embedding}}}. In \bibinfo{booktitle}{\emph{Proc.
  {{WSDM}}}}. \bibinfo{pages}{565--573}.
\newblock


\bibitem[Vaswani et~al\mbox{.}(2017)]%
        {Transformer}
\bibfield{author}{\bibinfo{person}{Ashish Vaswani}, \bibinfo{person}{Noam
  Shazeer}, \bibinfo{person}{Niki Parmar}, \bibinfo{person}{Jakob Uszkoreit},
  \bibinfo{person}{Llion Jones}, \bibinfo{person}{Aidan~N Gomez},
  \bibinfo{person}{{\L}ukasz Kaiser}, {and} \bibinfo{person}{Illia
  Polosukhin}.} \bibinfo{year}{2017}\natexlab{}.
\newblock \showarticletitle{Attention Is {{All}} You {{Need}}}. In
  \bibinfo{booktitle}{\emph{Proc. {{NeurIPS}}}}.
\newblock


\bibitem[Wei et~al\mbox{.}(2022)]%
        {weiMitigatingNeuralNetwork2022a}
\bibfield{author}{\bibinfo{person}{Hongxin Wei}, \bibinfo{person}{Renchunzi
  Xie}, \bibinfo{person}{Hao Cheng}, \bibinfo{person}{Lei Feng},
  \bibinfo{person}{Bo An}, {and} \bibinfo{person}{Yixuan Li}.}
  \bibinfo{year}{2022}\natexlab{}.
\newblock \showarticletitle{Mitigating {{Neural Network Overconfidence}} with
  {{Logit Normalization}}}. In \bibinfo{booktitle}{\emph{Proc. {{ICML}}}}.
\newblock
\showISSN{2640-3498}


\bibitem[Weston et~al\mbox{.}(2011)]%
        {westonWsabieScalingLarge2011}
\bibfield{author}{\bibinfo{person}{Jason Weston}, \bibinfo{person}{Samy
  Bengio}, {and} \bibinfo{person}{Nicolas Usunier}.}
  \bibinfo{year}{2011}\natexlab{}.
\newblock \showarticletitle{Wsabie: {{Scaling Up To Large Vocabulary Image
  Annotation}}}. In \bibinfo{booktitle}{\emph{Proc. {{IJCAI}}}}.
\newblock


\bibitem[Wu et~al\mbox{.}(2022)]%
        {wuEffectivenessSampledSoftmax2022}
\bibfield{author}{\bibinfo{person}{Jiancan Wu}, \bibinfo{person}{Xiang Wang},
  \bibinfo{person}{Xingyu Gao}, \bibinfo{person}{Jiawei Chen},
  \bibinfo{person}{Hongcheng Fu}, \bibinfo{person}{Tianyu Qiu}, {and}
  \bibinfo{person}{Xiangnan He}.} \bibinfo{year}{2022}\natexlab{}.
\newblock \bibinfo{title}{On the {{Effectiveness}} of {{Sampled Softmax Loss}}
  for {{Item Recommendation}}}.
\newblock
\newblock
\showeprint[arxiv]{2201.02327}~[cs]


\bibitem[Wu et~al\mbox{.}(2016)]%
        {wuGoogleNeuralMachine2016}
\bibfield{author}{\bibinfo{person}{Yonghui Wu}, \bibinfo{person}{Mike
  Schuster}, \bibinfo{person}{Zhifeng Chen}, \bibinfo{person}{Quoc~V. Le},
  \bibinfo{person}{Mohammad Norouzi}, \bibinfo{person}{Wolfgang Macherey},
  \bibinfo{person}{Maxim Krikun}, \bibinfo{person}{Yuan Cao},
  \bibinfo{person}{Qin Gao}, \bibinfo{person}{Klaus Macherey},
  \bibinfo{person}{Jeff Klingner}, \bibinfo{person}{Apurva Shah},
  \bibinfo{person}{Melvin Johnson}, \bibinfo{person}{Xiaobing Liu},
  \bibinfo{person}{{\L}ukasz Kaiser}, \bibinfo{person}{Stephan Gouws},
  \bibinfo{person}{Yoshikiyo Kato}, \bibinfo{person}{Taku Kudo},
  \bibinfo{person}{Hideto Kazawa}, \bibinfo{person}{Keith Stevens},
  \bibinfo{person}{George Kurian}, \bibinfo{person}{Nishant Patil},
  \bibinfo{person}{Wei Wang}, \bibinfo{person}{Cliff Young},
  \bibinfo{person}{Jason Smith}, \bibinfo{person}{Jason Riesa},
  \bibinfo{person}{Alex Rudnick}, \bibinfo{person}{Oriol Vinyals},
  \bibinfo{person}{Greg Corrado}, \bibinfo{person}{Macduff Hughes}, {and}
  \bibinfo{person}{Jeffrey Dean}.} \bibinfo{year}{2016}\natexlab{}.
\newblock \bibinfo{title}{Google's {{Neural Machine Translation System}}:
  {{Bridging}} the {{Gap}} between {{Human}} and {{Machine Translation}}}.
\newblock
\newblock
\showeprint[arxiv]{1609.08144}~[cs]


\bibitem[Xie et~al\mbox{.}(2022)]%
        {xieContrastiveLearningSequential2022}
\bibfield{author}{\bibinfo{person}{Xu Xie}, \bibinfo{person}{Fei Sun},
  \bibinfo{person}{Zhaoyang Liu}, \bibinfo{person}{Shiwen Wu},
  \bibinfo{person}{Jinyang Gao}, \bibinfo{person}{Jiandong Zhang},
  \bibinfo{person}{Bolin Ding}, {and} \bibinfo{person}{Bin Cui}.}
  \bibinfo{year}{2022}\natexlab{}.
\newblock \showarticletitle{Contrastive {{Learning}} for {{Sequential
  Recommendation}}}. In \bibinfo{booktitle}{\emph{Proc. {{ICDE}}}}.
  \bibinfo{pages}{1259--1273}.
\newblock


\bibitem[Yuan et~al\mbox{.}(2016)]%
        {yuanLambdaFMLearningOptimal2016}
\bibfield{author}{\bibinfo{person}{Fajie Yuan}, \bibinfo{person}{Guibing Guo},
  \bibinfo{person}{Joemon~M. Jose}, \bibinfo{person}{Long Chen},
  \bibinfo{person}{Haitao Yu}, {and} \bibinfo{person}{Weinan Zhang}.}
  \bibinfo{year}{2016}\natexlab{}.
\newblock \showarticletitle{{{LambdaFM}}: {{Learning Optimal Ranking}} with
  {{Factorization Machines Using Lambda Surrogates}}}. In
  \bibinfo{booktitle}{\emph{Proc. {{CIKM}}}}. \bibinfo{pages}{227--236}.
\newblock


\bibitem[Yuan et~al\mbox{.}(2019)]%
        {yuanSimpleConvolutionalGenerative2019}
\bibfield{author}{\bibinfo{person}{Fajie Yuan}, \bibinfo{person}{Alexandros
  Karatzoglou}, \bibinfo{person}{Ioannis Arapakis}, \bibinfo{person}{Joemon~M.
  Jose}, {and} \bibinfo{person}{Xiangnan He}.} \bibinfo{year}{2019}\natexlab{}.
\newblock \showarticletitle{A {{Simple Convolutional Generative Network}} for
  {{Next Item Recommendation}}}. In \bibinfo{booktitle}{\emph{Proc. {{WSDM}}}}.
  \bibinfo{pages}{582--590}.
\newblock


\bibitem[Zhou et~al\mbox{.}(2020)]%
        {zhouS3RecSelfSupervisedLearning2020}
\bibfield{author}{\bibinfo{person}{Kun Zhou}, \bibinfo{person}{Hui Wang},
  \bibinfo{person}{Wayne~Xin Zhao}, \bibinfo{person}{Yutao Zhu},
  \bibinfo{person}{Sirui Wang}, \bibinfo{person}{Fuzheng Zhang},
  \bibinfo{person}{Zhongyuan Wang}, {and} \bibinfo{person}{Ji-Rong Wen}.}
  \bibinfo{year}{2020}\natexlab{}.
\newblock \showarticletitle{S3-{{Rec}}: {{Self-Supervised Learning}} for
  {{Sequential Recommendation}} with {{Mutual Information Maximization}}}. In
  \bibinfo{booktitle}{\emph{Proc. {{CIKM}}}}. \bibinfo{pages}{1893--1902}.
\newblock


\end{thebibliography}
   
\end{document}